\documentclass[12pt,a4paper]{article}

\usepackage{epsf}
\usepackage[utf8]{inputenc}
\usepackage{epsfig}
\usepackage{verbatim}
\usepackage{amsfonts}
\usepackage{amsmath}
\usepackage{amssymb}
\usepackage{amsthm}
\usepackage{setspace}
\usepackage{latexsym}
\usepackage{multirow}
\usepackage{pifont}
\usepackage{graphics}
\usepackage{graphicx}
\usepackage{geometry}
\usepackage{array}
\usepackage{float}
\usepackage{lscape}
\usepackage{color}
\usepackage{calligra}
\usepackage{dsfont}
\usepackage{tabularx}
\usepackage{dcolumn}
\usepackage[bottom]{footmisc}
\usepackage[round]{natbib}
\usepackage{enumerate}
\usepackage[hyphens]{url}
\usepackage[normalem]{ulem}
\usepackage{epstopdf}
\usepackage[space]{grffile}
\usepackage{setspace}
\usepackage{amsmath}
\usepackage{amsthm}
\usepackage{amssymb}
\usepackage{supertabular}
\usepackage{subfig}
\def\sgn{\operatorname{sgn}}
\def\Nz{\mathbb{N}}
\def\0F1{{{}_0F_1}}
\def\al{\alpha}
\def\d{\partial}

\def\al{\alpha}

\def\Nz{\mathbb N}
\def\Zz{\mathbb Z}
\def\Rz{\mathbb R}
\def\Cz{\mathbb C}

\newtheorem{theorem}{Theorem}
\newtheorem{defi}{Definition}
\newtheorem{lem}{Lemma}
\newtheorem{ex}{Example}
\newtheorem{pr}{Proposition}

\newtheorem{co}{Corollary}

\usepackage{accents}

\usepackage{mathtools}

\DeclareMathAlphabet{\mathpzc}{OT1}{pzc}{m}{it}
\DeclareMathAlphabet{\mathcalligra}{T1}{calligra}{m}{n}
\begin{document}

\begin{titlepage}

 \title{No Midcost Democracy\thanks{We are grateful to Wolfgang Leininger for triggering this inquiry. We are also grateful to David K. Levine, Antoine Loeper, César Martinelli, Clemens Puppe, and the participants at the $5^{th}$ ETH Workshop on Political Economy, very especially our discussant Alberto Grillo, for valuable discussions. All errors are ours.}}

	\author{Hans Gersbach \\
		 \small KOF Swiss Economic Institute\\ [-2.0mm]
		 \small  ETH Zurich and CEPR\\  [-2.0mm]
		 \small  Leonhardstrasse 21\\ [-2.0mm]
		 \small 8092 Zurich, Switzerland\\  [-2.0mm]
		 \small hgersbach@ethz.ch \\
		\and
	\vspace{0.6mm}\\
    \and Arthur Schichl \\
        \small KOF Konjunkturforschungsstelle \\ [-2.0mm]
        \small ETH Zürich \\ [-2.0mm]
        \small Leonhardstrasse 21 \\ [-2.0mm]
		\small 8092 Zurich, Switzerland \\ [-2.0mm]
        \small aschichl@ethz.ch \\ 
        \and Oriol Tejada\footnote{Oriol Tejada acknowledges financial support from the Spanish Ministry of Science and Innovation within the call \textit{``Proyectos de Generaci\'on de Conocimiento 2021''}, PID2021-123747NA-100.}  \\ 
	\small Faculty of Economics and Business \\[-2.0mm]
	\small Universitat de Barcelona and BEAT \\[-2.0mm]
	\small Diagonal 690-696 \\[-2.0mm]
	\small 08034 Barcelona, Spain \\[-2.0mm]
	\small oriol.tejada@ub.edu \\
	}
	
\vspace{-1.5cm}
	\date{This version:  July 2025}
	\maketitle\thispagestyle{empty}

\vspace{-1.2cm}

    \begin{abstract}
        
        Which level of voting costs is optimal in a democracy? This paper argues that intermediate voting costs—what we term a "Midcost democracy"—should be avoided, as they fail to ensure that electoral outcomes reflect the preferences of the majority. We study a standard binary majority decision in which a majority of the electorate prefers alternative $A$ over alternative $B$. The population consists of partisan voters, who always participate, and non-partisan voters, who vote only when they believe their participation could be pivotal, given that voting entails a cost. We show that the probability of the majority-preferred alternative $A$ winning is non-monotonic in the level of voting costs. Specifically, when voting costs are either high or negligible, alternative $A$ wins in all equilibria. However, at intermediate cost levels, this alignment breaks down. These findings suggest that democratic systems should avoid institutional arrangements that lead to moderate voting costs, as they may undermine the majority principle.

\medskip 


\medskip


        
    \end{abstract}

\end{titlepage}

\onehalfspacing

\newpage


\section{Introduction}


The cost of voting plays a central role in the functioning of democratic systems. While it is partly determined by individuals’ choices and characteristics, it is also influenced by institutional design, which is the focus of this paper. Governments, or more generally, designers of social voting schemes, can shape voting costs through various mechanisms:

\begin{itemize}
    \item simplifying the issues at stake and improving the clarity and accessibility of explanatory materials,
    \item easing or eliminating voter registration requirements,
    \item investing in comprehensive voter education initiatives,
    \item expanding the number of polling stations and increasing the availability of voting booths and personnel to reduce wait times and travel distances (for in-person voting),
    \item streamlining the vote-by-mail process by including pre-addressed, pre-stamped return envelopes and designing ballots with a clear, user-friendly layout (for mail voting),
    \item simplifying the online voting procedure and offering an efficient helpline to support voters throughout the process (for online voting).
\end{itemize}

Electronic voting, in particular, can substantially decrease the costs of voting, since traveling and waiting time can be saved. Yet, the costs of registering, handling the electronic system, and the security credentials associated with electronic voting still have to be incurred. That is, voting costs can never be eliminated completely, no matter how voting takes place.

In this paper, we ask whether, and if so under which circumstances, it is desirable to lower the costs of voting to the lowest  feasible level. To address this question, we use a  standard model of costly voting in which a polity decides between two alternatives, and we make two assumptions. First, lowering the costs of voting is possible, but not down to zero. Ultimately, some time is needed to inform oneself, perform the act of voting, and to fill out the ballot. Second, a fraction of voters always votes, independent of the costs. The literature has identified various reasons why such a group of voters participates in elections despite the costs. For instance, this could be due to social norms or ethical concerns developed in seminal papers by \citep[]{LevineMattozzi2020,feddersen,feddersen2006ethical,CoateBesley1998}.

Given the above setup, we examine at which level a social voting procedure designer who wants to maximize ex ante utilitarian welfare  should set the costs of voting for all citizens.  Utilitarian welfare takes into account the probability that the alternative with higher ex ante support is implemented. Also, the expected aggregate costs of voting matter for utilitarian welfare. Yet, in practice, the first objective typically dominates the second one as voting costs, even in the aggregate, are considerably less important than the consequences of policy decisions. Accordingly, we 
assume that the social voting procedure designer aims at choosing the level of costs as a mere instrument, namely to maximize the probability that the alternative wins which generates higher expected utilitarian welfare.\footnote{We discuss in the concluding section how voting costs can be taken into account in the design of optimal voting procedures.}  

For this purpose, we consider a standard binary majority decision. A majority in the polity favors one alternative ($A$) over the other alternative ($B$). Some citizens are partisan voters and always turn out to vote, while the other citizens only vote if they perceive that their votes might make a difference, as voting is costly. The main task is to set the voting cost at levels at which alternative $A$ wins with the highest possible probability, and if possible, with probability one.

Our main results are as follows. First, we characterize all possible type--symmetric Bayesian Nash equilibria (equilibria, in short) for each level of voting costs. In a type--symmetric equilibrium, players with the same preferences play the same strategy. It turns out that such a characterization requires particular mathematical tools and reveals different types of equilibria that occur, which render our analysis challenging.  For a given level of voting costs, we have  unique or multiple equilibria. Below we provide a detailed summary of these formal results.

Second, we identify an intermediate range of costs for which it is not guaranteed that alternative $A$ wins and thus the  will of the majority in the polity fails to be reflected in the voting outcome. In particular, we establish that the likelihood that alternative $A$ wins is non-monotonic in the cost of voting. As a consequence, intermediate voting costs in a democracy, i.e., what we call the ``Midcost democracy",   should be avoided.\footnote{If the minimum feasible voting costs falls in the region of intermediate costs, then such a minimum feasible cost should not be chosen and a greater cost should be implemented instead.} 

Third, we show that all equilibria can be recovered in alternative versions of the game. In one version, we consider a game in which players select only pure strategies, but we drop type symmetry. In the other version, we assume that two parties compete and can coordinate how many supporters are going to vote.

It is useful to provide a summary of all possible types of  the different type-symmetric equilibria that can exist for particular cost levels. In a so-called coin-toss equilibrium, all non-partisan voters play mixed strategies and both alternatives receive the same votes in expectation and have a chance of one half to win. 

In all other equilibria, the majority in the electorate also wins the voting outcome. There are five types of such equilibria.
There are two types of so-called partial--absenteeism equilibria if either all non-partisan $A$--supporters or all non-partisan $B$--supporters do not participate in voting, respectively. A special case of the above is when both non-partisan $A$--supporters and non-partisan $B$--supporters do not participate and the outcome is solely determined by the partisan voters who always vote. This equilibrium is called the no--queue equilibrium.

Symmetrically, if either all non-partisan $A$--supporters or all non-partisan $B$--supporters vote, we obtain two types of equilibria which are called partial--saturation equilibria.
A special case occurs when all non-partisan voters in the polity and thus all voters cast their vote. This is called an all--swipe equilibrium. 

Hence, a voting procedure designer should set the voting costs at a level at which the coin-toss equilibrium does not occur and at least one of the other types of equilibria does occur. It will turn out that at least one type of the equilibria, in which alternative $A$ wins and thus the will of the majority prevails, exists for any cost level and thus the task is simply to avoid setting costs for which the coin-toss equilibria do exist.  

Incorporating voting costs into voting theory goes back at least to \cite{RikerOrdeshook1968} and the later works of \cite{palfrey1983strategic} and \cite{ledyard1984pure}. The theory of costly voting has seen some significant advancement ever since (see \cite{borgers}). On the one hand, it has been applied to institutional settings and small-sized electorates. On the other hand, in the limit of very large electorates, \cite{herrera_stuff} consider a Poisson voting game, where the individual voting costs are drawn from a distribution. \cite{polborn} and \cite{xefteris} study costly voting with three alternatives.

 Our paper is normative in that we aim to determine a main variable of the above models, the voting cost(s). There exist previous attempts at the problem of voting design which focused on other aspects of elections.
Subsidizing or forcing a subset of voters to vote in a costly voting setting is studied in~\cite{AV}. Modifying voting costs through social norms is analyzed in~\cite{LevineMattozzi2020}. Although for simplicity we focus on a setup in which all voters incur the same voting costs, similar insights would be obtained if we assumed heterogeneous costs for voters drawn independently from the same non-degenerate distribution~\citep[see e.g.][]{herrera_stuff}.

The paper is organized as follows. In Section~\ref{sec:model} we introduce our baseline setting and set up the notation. In Section \ref{sec:tools} we introduce the mathematical tools to analyze the equilibrium conditions. In Section \ref{sec:analysis} we characterize all equilibria of the voting game.  In Section \ref{sec:equivalent}, we show that the equilibria can be recovered in an alternative voting game. Section~\ref{sec:conclusion} concludes. The main proofs are in the appendix.

\section{Model}\label{sec:model}

\subsection{The Set-up}

A society must choose one of two alternatives, denoted by~$A$ and~$B$. The total number of voters is denoted by $N \in \mathbb{N}$. 

In the society, there are two groups of voters. One group always votes and thus is willing to incur any cost of voting (these citizens are henceforth called ``partisan voters''), either because these voters follow social norms, want to behave ethically, or are incentivized to vote by being given a subsidy equal to their cost of voting. The other group of citizens ponders whether to vote and to incur the cost of voting (they are called ``non-partisan voters''). We denote by~$p$  the proportion of partisan voters with respect to the total number of voters. Thus, $1-p$ is the proportion of voters that are non-partisans. The cost of voting comprises the components mentioned in the Introduction, most notably the time and effort needed to vote. We assume the value of~$p$ to be common knowledge, i.e., in most of our results later on, $p$ is fixed.

We denote the proportion of $A$-supporters with respect to the total number of voters by $p_A$. Then the proportion of $B$-supporters is $p_B := 1 - p_A$. We also assume that all events (partisanship and preference) are independent across all voters.  Without loss of generality, we assume that alternative $A$ is the ex ante favorite, i.e., $p_A>\frac{1}{2}.$ Preferences are private information, while the values of~$p_A$ and~$p_B$ are common knowledge, i.e., again, for most of our results, $p_A$ is fixed.

Since we are concerned with large populations, we assume that the exact size of the electorate is unknown to the voters, although they have sufficient information to give a good estimate of $N$. Quantitatively, in our model, the voter assumes that the total number of voters at the moment of the election is distributed like a Poisson random variable with parameter~$N$. For convenience, we then use the following notation  throughout the paper, for any $m\in \mathbb{N}$ and $\lambda>0$:
\begin{equation*}
    P_x(m):=\frac{x^m}{e^x\cdot m!}.
\end{equation*}
That is, $P_x(m)$ denotes the probability that a Poisson random variable with parameter (or expected value)~$x$ is equal to~$m$. Furthermore, we assume the voter to be aware of the fact  that partisanship and preference are independent from each other and the size of the citizenship. Hence, the voter  deduces from the decomposition theorem \citep{poisson} that the number of partisan voters $x_{par}$ and the number of non--partisan voters $x_{non}$ are also independent Poisson distributed random variables with parameters $Np$ and $N(1-p)$, respectively. Finally, the political preference of a citizen is independent of the population and partisanship. Therefore, again, the decomposition theorem yields that $x_{par}^A$, $x_{par}^B$, $x_{non}^A$ and $x_{non}^B$---the number of partisan $A$--voters, partisan $B$--voters, non--partisan $A$--voters and non--partisan $B$--voters respectively---are again independent Poisson distributed random variables with parameters $x_A := Npp_A$, $x_B := Np(1-p_A)$, $N(1-p)p_A$, and $N(1-p)(1-p_A)$. 

The citizens' utilities are as follows: No matter whether they are partisan or not, if a citizen's preferred alternative is chosen, the citizen receives a utility gain which is normalized to one, while the utility is normalized to zero if the least preferred alternative is chosen by the citizenry. Besides this utility component, a voter incurs an additional utility $-c<0$ when casting a vote, so $c>0$ can be seen as the voting cost.

To analyze how voting costs should be set, we model the voting process as a (static) game $\mathcal{G}(c)$, which we shall just denote by $\mathcal{G}$ for simplicity. The players of the game are the $N(1-p)$ non-partisan voters, who anticipate that partisan voters always vote. We order them such that the $N(1-p)p_A$ non--partisan $A$--supporters are indexed in $I_A \coloneqq [\![1, N(p-1)p_A]\!]$ and the $N(1-p)(1-p_A)$ non--partisan $B$--supporters are indexed in $I_B \coloneqq [\![N(p-1)p_A+1, N(1-p)]\!]$. Their mixed strategy set is $[0,1]$, where the two pure strategies, namely voting for their preferred alternative or abstaining, are respectively denoted by $1$ and $0$. We denote the $i$--th player's strategy by $\sigma_i$. With two alternatives $A$ and $B$, voting for the preferred alternative weakly dominates voting for the other alternative, and thus we implicitly assume that (non-partisan) voters do not use weakly dominated strategies.  

To define the utility of the players, it comes in handy to introduce the following notation for any $m,n\in \mathbb{N}$:
\begin{align*}
    f(m,n):=\begin{cases}
               1 & \text{if } m>n, \\
               \frac{1}{2} & \text{if } m=n,\\
               0 & \text{if } m<n.
            \end{cases}    
\end{align*}
The function $f$ denotes the expected gain for $A$-supporters (partisan and non-partisan) if alternative $A$ receives $m$ votes and alternative $B$ receives $n$ votes. If there is a tie $(m = n)$, we assume that each alternative is chosen with probability~$\frac{1}{2}$. 

Now, we will formally define the player's utility functions. As explained above, from the point of view of the players, the voting game is a Poisson game and the corresponding (perceived) utilities are, thus, as follows:

Let $\alpha_A,\alpha_B \in [0,1]$. We define the functions $u_A: \{0,1\} \to \mathbb{R}$ and $u_B: \{0,1\} \to \mathbb{R}$ by
\begin{align*}
&u_A(0,\alpha_A,\alpha_B) \\&\qquad\coloneqq  \sum_{a=0}^{\infty}\sum_{b=0}^{\infty}P_{Npp_A}(a)P_{Np(1-p_A)}(b)\sum_{r=0}^{\infty}\sum_{s=0}^{\infty}P_{N(1-p)p_A\alpha_A}(r)P_{N(1-p)(1-p_A)\alpha_B}(s)f(a+r, b+s)\\
&u_A(1,\alpha_A,\alpha_B)\\ &\qquad\coloneqq \sum_{a=0}^{\infty}\sum_{b=0}^{\infty}P_{Npp_A}(a)P_{Np(1-p_A)}(b)\sum_{r=0}^{\infty}\sum_{s=0}^{\infty}P_{N(1-p)p_A\alpha_A}(r)P_{N(1-p)(1-p_A)\alpha_B}(s)f(a+r+1, b+s) - c
\end{align*}
and
\begin{align*}
&u_B(0,\alpha_A,\alpha_B)\\ &\qquad\coloneqq  \sum_{a=0}^{\infty}\sum_{b=0}^{\infty}P_{Npp_A}(a)P_{Np(1-p_A)}(b)\sum_{r=0}^{\infty}\sum_{s=0}^{\infty}P_{N(1-p)p_A\alpha_A}(r)P_{N(1-p)(1-p_A)\alpha_B}(s)f(b+s, a+r)\\
&u_B(1,\alpha_A,\alpha_B) \\&\qquad\coloneqq  \sum_{a=0}^{\infty}\sum_{b=0}^{\infty}P_{Npp_A}(a)P_{Np(1-p_A)}(b)\sum_{r=0}^{\infty}\sum_{s=0}^{\infty}P_{N(1-p)p_A\alpha_A}(r)P_{N(1-p)(1-p_A)\alpha_B}(s)f(b + s + 1, a + r) - c.
\end{align*}

The expression of $u_A(1,\alpha_A,\alpha_B)$, for example, is the perceived utility of an $A$--supporter who votes, given that the expected numbers of non--partisan $A$--supporters and  non--partisan $B$--supporters who vote are $N(1-p)p_A\al_A$ and $N(1-p)(1-p_A)\al_B$ respectively. If we let $E_{a,b,r,s}$ denote the event in which the numbers of partisan $A$--supporters, partisan $B$--supporters, non--partisan $A$--supporters and non--partisan $B$--supporters---other than the player itself---who vote are $a$, $b$, $r$, and $s$, respectively, then $u_A(1,\alpha_A,\alpha_B)$ is just
\begin{equation}
    \sum_{a=0,b=0,r=0,s=0}^{\infty}\mathbb{P}(E_{a,b,r,s})\bigl(f(a+r+1, b+s) - c\bigr).
\end{equation}
In other words, we have a sum over $a,b,r,s$ of the utilities the voter obtains if the event $E_{a,b,r,s}$ occurs, i.e., $f(a+r+1, b+s) - c$, times the probabilities of said events. Note that, since the numbers of partisan $A$--supporters, partisan $B$--supporters, non--partisan $A$--supporters, and non--partisan $B$--supporters who vote are supposed to be independent, the probability of $E_{a,b,r,s}$ is $P_{Npp_A}(a)P_{Np(1-p_A)}(b)P_{N(1-p)p_A\alpha_A}(r)P_{N(1-p)(1-p_A)\alpha_B}(s)$. 

The other expression, $u_A(0,\alpha_A,\alpha_B)$, i.e., the perceived utility of an $A$--supporter who does not vote, given that the expected numbers of non--partisan $A$--supporters and  non--partisan $B$--supporters who vote are $N(1-p)p_A\al_A$ and $N(1-p)(1-p_A)\al_B$ respectively, is obtained in a similar fashion.

Fix $i \in [\![1, N(p-1)]\!]$ and strategies $\sigma_j$ for all $j \neq i$. Then, if $i \in I_A$, the $i$--th player's utility is given by 
\begin{align}\label{eq:utilityA}
    U_i(\sigma_i,(\sigma_j)_{j \neq i}) &\coloneqq (1-\sigma_i) u_A\biggl(0, \frac{\sum_{j \neq i, j \in I_A}\sigma_j}{N(1-p)p_A - 1}, \frac{\sum_{j \in I_B}\sigma_j}{N(1-p)(1-p_A)}\biggr) \\
    &+ \sigma_i u_A\biggl(1, \frac{\sum_{j \neq i, j \in I_A}\sigma_j}{N(1-p)p_A - 1}, \frac{\sum_{j \in I_B}\sigma_j}{N(1-p)(1-p_A)}\biggr)\notag
\end{align}
and 
if $i \in I_B$, the $i$--th player's utility is given by 
\begin{align}
    U_i(\sigma_i,(\sigma_j)_{j \neq i}) &\coloneqq (1-\sigma_i) u_B\biggl(0, \frac{\sum_{ j \in I_A}\sigma_j}{N(1-p)p_A}, \frac{\sum_{j \neq i,j \in I_B}\sigma_j}{N(1-p)(1-p_A) - 1}\biggr) \\
    &+ \sigma_i u_B\biggl(1, \frac{\sum_{j \in I_A}\sigma_j}{N(1-p)p_A}, \frac{\sum_{j \neq i,j \in I_B}\sigma_j}{N(1-p)(1-p_A) - 1}\biggr).\notag
\end{align}
Since $U_i$ only depends on the party the $i$--th voter supports, we shall just write $U_A(\sigma_i,(\sigma_j)_{j \neq i}) \coloneqq  U_i(\sigma_i,(\sigma_j)_{j \neq i})$, if the $i$--th player is an $A$--supporter, and $U_B(\sigma_i,(\sigma_j)_{j \neq i}) \coloneqq  U_i(\sigma_i,(\sigma_j)_{j \neq i})$, if the $i$--th player is a $B$--supporter. Furthermore, we write 
\begin{equation}
    \begin{aligned}
        \alpha_A &\coloneqq \frac{\sum_{ j \in I_A}\sigma_j}{N(1-p)p_A},\notag\\
        \bar{\alpha}_A^i &\coloneqq \frac{\sum_{j \neq i, j \in I_A}\sigma_j}{N(1-p)p_A - 1},\notag\\
        \alpha_B &\coloneqq \frac{\sum_{j \in I_B}\sigma_j}{N(1-p)(1-p_A)},\notag\notag\\
        \bar{\alpha}_B^i &\coloneqq \frac{\sum_{j \neq i,j \in I_B}\sigma_j}{N(1-p)(1-p_A) - 1}.\notag
    \end{aligned}
\end{equation}

The idea behind utility functions $U_A$ and $U_B$ is that, although the number of players of each category (partisanship and preference) is fixed, these numbers are unknown to the players. An $A$--supporter, for example, perceives the numbers of partisan $A$--voters $X_A$, partisan $B$--voters $X_B$, non--partisan $A$--voters $Y_A$, and non--partisan $B$--voters $Y_B$ (other than themselves) as independent and Poisson distributed random variables, with respective parameters $x_A \coloneqq Npp_A$, $x_B \coloneqq Np(1-p_A)$, $y_A:=N(1-p)p_A\bar{\alpha}_A^i$ and $y_B := N(1-p) (1-p_A) \alpha_B$. The expression (\ref{eq:utilityA}) is actually nothing more than
\begin{equation*}
\mathbb{E}[\sigma_i\bigl(f (X_A + Y_A + 1, X_B + Y_B) - c\bigr) + (1-\sigma_i)f (X_A + Y_A, X_B + Y_B)],
\end{equation*}
where the random variables $X_A,X_B,Y_A,Y_B$ are distributed as described above.

Note that to derive the above conditions, we have used the  environmental equivalence property \citep{poisson}, which states that any player of the game, once s/he finds out that s/he is a player of the game, forms the same beliefs about the number of other players as an outside observer.

We then study the  type-symmetric Bayesian Nash equilibria (equilibria, in short) of  voting game~$\mathcal{G}$, either in pure or mixed strategies. By type-symmetric equilibria, we mean  that all  non-partisan voters supporting the same alternative  use the same strategy.  We let~$\alpha_A$  denote the probability that a (non-partisan) $A$-supporter votes and $\alpha_B$ denote the probability that a (non-partisan) $B$-supporter votes. Then an equilibrium of $\mathcal{G}$ is a pair $(\alpha_A,\alpha_B) in [0,1]\times [0,1]$.


\subsection{Equilibrium Conditions and Types of Equilibria}

Let us describe the different type--symmetric equilibria that can occur in game $\mathcal{G}$. The utility functions $q \mapsto U_A(q)$ and $q \mapsto U_B(q)$ are linear in $q$, hence they are either constant, or attain their unique maximum in either $0$ or $1$.

Suppose first that $(\alpha_A,\alpha_B)$ is a type--symmetric equilibrium for $\alpha_A,\alpha_B \in ]0,1[$. Then, $q \mapsto U_A(q) = (1-q)u_A(0, \alpha_A,\alpha_B) + qu_A(1, \alpha_A,\alpha_B)$ and $q \mapsto U_B(q) = (1-q)u_B(0, \alpha_A,\alpha_B) +  qu_B(1, \alpha_A,\alpha_B)$ attain their maxima in $\alpha_A$ and $\alpha_B$ respectively. Indeed, in a type-symmetric equilibrium, all $A$--supporters play the same strategy $q_A$, and, thus, $\alpha_A = \bar{\alpha}_A^i = q_A$. A similar reasoning for $B$--supporters shows that $\alpha_A$ and $\alpha_B$ must be best responses. Hence, since $U_A$ and $U_B$ are both linear, they are constant, i.e., the indifference conditions
\begin{equation}\label{A_supporter}
    c = \sum_{a=0}^{\infty}\sum_{b=0}^{\infty}P_{x_A}(a)P_{x_B}(b)\sum_{r=0}^{\infty}\sum_{s=0}^{\infty}P_{y_A}(r)P_{y_B}(s)(f(a+r+1, b+s) - f(a+r, b+s))
\end{equation}
of an $A$--supporter and
\begin{equation}\label{B_supporter}
    c = \sum_{a=0}^{\infty}\sum_{b=0}^{\infty}P_{x_A}(a)P_{x_B}(b)\sum_{r=0}^{\infty}\sum_{s=0}^{\infty}P_{y_A}(r)P_{y_B}(s)(f(b + s + 1, a + r) - f(b+s, a+r))
\end{equation}
of a $B$--supporter need to be satisfied. We call an equilibrium of this form a coin--toss equilibrium. It will turn out that the expected numbers of voters for each of the two alternatives are equal in a coin--toss equilibrium. It then depends on arbitrarily small differences in the turnout which alternative wins and the probability that a given alternative has a higher turnout of supporters is $\frac{1}{2}$. 

Now, suppose, for example that $\alpha_A = 0$ in a type--symmetric equilibrium, i.e., all $A$--supporters do not participate in the vote. Then, $q \mapsto U_A(q) = (1-q)u_A(0, 0,\alpha_B) + qu_A(1, 0,\alpha_B)$ must attain a maximum in $0$. Unlike before, this does not require $U_A$ to be constant, but $U_A$ to be decreasing in $q$. In other words, the indifference condition of an $A$--supporter (\ref{A_supporter}) is replaced by the border condition
\begin{equation}\label{A_supporter_partmix}
    c \geq \sum_{a=0}^{\infty}\sum_{b=0}^{\infty}P_{x_A}(a)P_{x_B}(b)\sum_{r=0}^{\infty}\sum_{s=0}^{\infty}P_{y_A}(r)P_{y_B}(s)(f(a+r+1, b+s) - f(a+r, b+s)).
\end{equation}
Analogously, if $\alpha_B = 0$
the indifference condition of a $B$--supporter (\ref{B_supporter}) is replaced by the border condition
\begin{equation}\label{B_supporter_partmix}
    c \geq \sum_{a=0}^{\infty}\sum_{b=0}^{\infty}P_{x_A}(a)P_{x_B}(b)\sum_{r=0}^{\infty}\sum_{s=0}^{\infty}P_{y_A}(r)P_{y_B}(s)(f(b + s + 1, a + r) - f(b+s, a+r)).
\end{equation}
We call such equilibria partial--absenteeism equilibria.

A special case of the above is $(\alpha_A, \alpha_B) = (0,0)$. If both (\ref{A_supporter_partmix}) and (\ref{B_supporter_partmix}) are satisfied in this case, we say that $(0,0)$ is the no--queue equilibrium.

Symmetrically, if either $\alpha_A$ or $\alpha_B$ or both of them are equal to $1$, the indifference conditions (\ref{A_supporter}) and (\ref{B_supporter}) are respectively replaced by the border conditions
\begin{equation}\label{A_supporter_partmix_sat}
    c \leq \sum_{a=0}^{\infty}\sum_{b=0}^{\infty}P_{x_A}(a)P_{x_B}(b)\sum_{r=0}^{\infty}\sum_{s=0}^{\infty}P_{y_A}(r)P_{y_B}(s)(f(a+r+1, b+s) - f(a+r, b+s))
\end{equation}
and
\begin{equation}\label{B_supporter_partmix_sat}
    c \leq \sum_{a=0}^{\infty}\sum_{b=0}^{\infty}P_{x_A}(a)P_{x_B}(b)\sum_{r=0}^{\infty}\sum_{s=0}^{\infty}P_{y_A}(r)P_{y_B}(s)(f(b + s + 1, a + r) - f(b+s, a+r)).
\end{equation}
We call such equilibria partial--saturation equilibria.

A special case of the above is $(\alpha_A, \alpha_B) = (1,1)$. If both (\ref{A_supporter_partmix_sat}) and (\ref{B_supporter_partmix_sat}) are satisfied in this case, we say that $(1,1)$ is the all--swipe equilibrium.

\subsection{The winning party}

In a voting game with pure strategies, it is trivial to determine the winning party---everything comes down to counting the respective numbers of $A$-- and $B$--supporters who chose the strategy voting over the strategy abstaining. The situation becomes more complex in a mixed strategy game, as the number of players who vote is probabilistic, and, thus, so is the winner of the election.

Let $(\alpha_A,\alpha_B)$ be a type-symmetric equilibrium of $\mathcal{G}$. In this equilibrium, the $N(1-p)p_A$ non--partisan $A$--supporters all choose the strategy voting with probability $\alpha_A$ and the $N(1-p)(1-p_A)$ non--partisan $A$--supporters all choose the strategy voting with probability $\alpha_B$. Hence, the number of non--partisan $A$--supporters who vote is distributed like a binomial random variable $Z_A\sim B(N(1-p)p_A,\alpha_A)$ and  the number of non--partisan $B$--supporters who vote is distributed like a binomial random variable $Z_B\sim B(N(1-p)(1-p_A),\alpha_B)$. So, the expected values of non--partisan $A$-- and $B$--supporters who vote are respectively $e_A \coloneqq N(1-p)p_A\alpha_A$ and $e_B \coloneqq N(1-p)(1-p_A)\alpha_B$ and, thus, the expected number of $A$-- and $B$--voters are respectively $Np_A(p + (1-p)\alpha_A)$ and $N(1-p_A)(p + (1-p)\alpha_B)$. Therefore, we would ``expect'' party $A$ to win if 
\begin{equation}\label{eq:wincond}
Np_A(p + (1-p)\alpha_A) > N(1-p_A)(p + (1-p)\alpha_B).
\end{equation}

That the standard deviations of $Z_A$ and $Z_B$ are given by $\sigma_A \coloneqq \sqrt{N}\sqrt{(1-p)p_A\alpha_A(1-\alpha_A)}$ and $\sigma_B \coloneqq \sqrt{N}\sqrt{(1-p)(1-p_A)\alpha_B(1-\alpha_B)}$ respectively. We have $\frac{\sigma_A}{e_A} \sim_{N \to \infty} \frac{1}{\sqrt{N}}$. This means, in some sense, that for $N$ large enough, the respective numbers of $A$--voters and $B$--voters almost show a deterministic behavior, i.e., the number of voters is almost equal to the expected number of voters with a very high probability. Therefore, since we are mainly interested in the voting behavior of large populations, we shall write from now on, that party $A$ is the winner of the election if condition (\ref{eq:wincond}) holds.

Nevertheless, we stress that, even if condition (\ref{eq:wincond}) holds, party $B$ might still win the election, albeit with a small likelihood---given that $Npp_A < N(1-pA)$. Just in the cases of the no--queue and all--swipe equilibria, the numbers of voters for both parties are completely deterministic, such that party $A$ always wins. Note that this is no logical contradiction to the existence of an all--swipe equilibrium, as the $B$--supporters in our model are unaware of their certain loss---recall that in their subjective perception, the numbers of players and supporters of each party are themselves probabilistic. 

\section{Mathematical Tools and Simplifications}\label{sec:tools}

\subsection{Mathematical Tools}

In this subsection, we introduce the mathematical tools to examine the equilibrium conditions. 
Before we  introduce these tools, let us recall that $O(f(m))$ denotes the class of functions that are asymptotically dominated by $f(m)$. Formally, for any $g(m)\in O(f(m))$, we have $\lim_{m\rightarrow \infty}\frac{|g(m)|}{|f(m)|}<\infty$. Furthermore, $o(f(m))$ denotes the class of functions that are asymptotically strictly dominated by $f(m)$. Formally, for any $g(m)\in o(f(m))$, we have $\lim_{m\rightarrow \infty}\frac{|g(m)|}{|f(m)|}=0$.

 We next introduce the definitions of main functions we will use. We start with the definition
of a generalized hypergeometric function. 

\begin{defi}\label{de:hypgeom}
Let $a \in \Rz \setminus \Zz_-$. The generalized hypergeometric function $\0F1(;a,z)$ is defined by 
$$
	\0F1(;a,z) := \sum_{k=0}^\infty \frac1{(a)_k}\frac{z^k}{k!},
$$
where $(a)_k$ is the rising factorial Pochhammer symbol, which is recursively defined as
$$
	\begin{aligned}
		(a)_0 &= 1 \\
		(a)_n &= (a)_{n-1}(a+n-1).
	\end{aligned}
$$
For integers $a$, we have $(a)_n = \tfrac{(a+n-1)!}{(a-1)!}$, with the special case $(1)_n = n!$.
\end{defi} 
We will need the following special cases:
\begin{ex}
In particular, 
\begin{equation}
	\begin{aligned}
		\0F1(;1,z) &= \sum_{k=0}^\infty \frac{z^k}{(k!)^2} \\
		\0F1(;2,z) &= \sum_{k=0}^\infty \frac{z^k}{k!(k+1)!} \\
		\0F1(;3,z) &= 2\sum_{k=0}^\infty \frac{z^k}{k!(k+2)!}.
	\end{aligned}
\end{equation}
\end{ex}
These functions have first been introduced as solutions of the so-called hypergeometric differential equation. In particular, the differential equation satisfied by $\0F1(;1,z)$ is the following.
\begin{lem}\label{le:hgeomdiff}
We have $$\frac{\d}{\d z}\0F1(;1,z) = \0F1(;2,z).$$ It follows for all $z \in \Cz$,
\begin{equation}
z\frac{\d^2}{\d z^2}\0F1(;1,z) + \frac{\d}{\d z}\0F1(;1,z) - \0F1(;1,z) = 0.
\end{equation}
\begin{proof}
This is a straightforward computation---we derive the hypergeometric series term by term. First we note that, for all $z \in \Cz$, 
\begin{equation}
z\frac{\d}{\d z}\0F1(;1,z) = \sum_{k\geq1} \frac{1}{k!(k-1)!}z^k,
\end{equation}
and then, we take the derivative on both sides.
\end{proof}
\end{lem}
In order to simplify some of our further computations, we also introduce the so-called modified Bessel functions of first kind. 
\begin{defi}
Let $a \in \Rz \setminus \Zz^*_-$. The modified Bessel function of first kind and order $a$ is defined by
$$
	I_a(z) := \frac{\bigl(\tfrac z2\bigr)^a}{\Gamma(a+1)}\0F1(;a+1,\tfrac14 z^2).
$$
\end{defi}
We shall only need the Bessel functions of orders $0$ and $1$.
\begin{ex}
The modified Bessel functions of orders $0$ and $1$ are given by
$$
	\begin{aligned}
		I_0(z) &= \0F1(;1,\tfrac14 z^2)\\
		I_1(z) &= \tfrac z2\0F1(;2,\tfrac14 z^2).
	\end{aligned}
$$
\end{ex}
Lemma~\ref{le:hgeomdiff} immediately yields that $I_0$ satisfies the following differential equation.
\begin{co}\label{co:Bessdiff}
For all $z\in \Cz$, $$I_0(z)' = I_1(z)$$ and
\begin{equation}
z\biggl(\frac{\d^2}{\d z^2} I_0(z) - I_0(z)\biggr) = - \frac{\d}{\d z} I_0(z).
\end{equation}
\end{co}

Before proceeding, we also give an asymptotic expansion of $I_0$ and $I_1$ for $z \to +\infty$.

\begin{lem}\label{le:Bessasymp}
Asymptotically, the modified Bessel functions $I_0$ and $I_1$ behave like
\begin{equation}
\begin{aligned}
I_0(x) &=_{+\infty} \frac{e^x}{\sqrt{2\pi x}}\bigl(1 + \frac{1}{8x}\bigr) + O\bigl(\frac{e^x}{x^\frac{5}{2}}\bigr) \\
I_1(x) &=_{+\infty} \frac{e^x}{\sqrt{2\pi x}}\bigl(1 - \frac{3}{8x}\bigr) + O\bigl(\frac{e^x}{x^\frac{5}{2}}\bigr).
\end{aligned}
\end{equation}
\end{lem}

\begin{proof}
    This is a known result, see for instance Chapter 7 in \citep{watson1922treatise}.
\end{proof}

\subsection{Simplifying the Equilibrium Conditions}

Now, we show that the right--hand sides of (\ref{A_supporter}) and (\ref{B_supporter}) have particularly simple
expressions, involving only some of the---well--studied---hypergeometric functions defined in Definition~\ref{de:hypgeom}. We derive the following expression:
\begin{pr}\label{pr:rhsexpr}
Denote by $R_1$ and $R_2$ the right--hand sides of (\ref{A_supporter}) and (\ref{B_supporter}), respectively---as functions of $y_A$ and $y_B$. Then, for all $y_A, y_B \geq 0$,
\begin{align}
R_1(y_A,y_B) &= \tfrac12\bigl(\0F1(;1,(x_A+y_A)(x_B+y_B))\notag\\
		&\qquad\qquad+(x_B+y_B)\,\0F1(;2,(x_A+y_A)(x_B+y_B))\bigl)e^{-x_A-y_A-x_B-y_B} \label{eq:1D}\\
R_2(y_A,y_B) &=\tfrac12\bigl(\0F1(;1,(x_A+y_A)(x_B+y_B))\notag\\
		&\qquad\qquad+(x_A+y_A)\,\0F1(;2,(x_A+y_A)(x_B+y_B))\bigl)e^{-x_A-y_A-x_B-y_B}.\label{eq:2D}
\end{align}
Now, define 
\begin{equation}
\begin{aligned}
r_1(\al_A,\al_B) &:= R_1(N(1-p)p_A\al_A,N(1-p)(1-p_A)\al_B) \notag\\
r_2(\al_A,\al_B) &:= R_2(N(1-p)p_A\al_A,N(1-p)(1-p_A)\al_B).\notag
\end{aligned}
\end{equation} 
Then,
\begin{equation}
\begin{aligned}
r_1(\al_A,\al_B)&=\tfrac12\bigl(\0F1\bigl(;1,N^2p_A(1-p_A)(p+\al_A(1-p))(p+\al_B(1-p))\bigr)\notag\\
		&\qquad+N(1-p_A)(p+(1-p)\al_B)\notag\\
		&\qquad\qquad\cdot\0F1\bigl(;2,N^2p_A(1-p_A)(p+\al_A(1-p))(p+\al_B(1-p))\bigr)\bigl)\notag\\
		&\qquad\qquad\qquad\qquad\qquad\cdot e^{-N(p+(1-p)(p_A\al_A+(1-p_A)\al_B))}\notag\\
r_2(\al_A,\al_B)&=\tfrac12\bigl(\0F1\bigl(;1,N^2p_A(1-p_A)(p+\al_A(1-p))(p+\al_B(1-p))\bigr)\notag\\
		&\qquad+N p_A(p+(1-p)\al_A)\notag\\
		&\qquad\qquad\cdot\0F1\bigl(;2,N^2p_A(1-p_A)(p+\al_A(1-p))(p+\al_B(1-p))\bigr)\bigl)\notag\\
		&\qquad\qquad\qquad\qquad\qquad\cdot e^{-N(p+(1-p)(p_A\al_A+(1-p_A)\al_B))}.\notag
\end{aligned}
\end{equation}
\end{pr}

Having simplified the right--hand sides of (1) and (2), we can start to study the existence 
of equilibria depending on the voting costs $c$.

\section{Equilibrium Analysis}\label{sec:analysis}

\subsection{Coin--toss Equilibria}

We first look for coin--toss equilibria and obtain:

\begin{theorem}\label{th:mixed}
There exists at most one coin--toss equilibrium $(\al_A,\al_B)$. 
Now, define for all $z > 0$, $g(z):=(I_0(z)+I_1(z))e^{-z}$. There exists exactly one coin--toss equilibrium $(\al_A,\al_B)$ if and only if
\begin{equation}\label{eq:4x}
	g(2x_A) > 2c > g(2N(1-p_A)).
\end{equation}
In particular, $g(2x_A) \leq g(x_A + x_B) = g(Np)$, and $g(Np)$ is decreasing in $N$ with 
\begin{equation}
g(Np) =_{N \to +\infty} \sqrt{\frac{2}{\pi Np}} + O\biggl(\frac{1}{(Np)^\frac{3}{2}}\biggr).
\end{equation} 
Also the lower bound is decreasing and behaves like $\frac{1}{\sqrt{N}}$ for $N \to +\infty$.
Furthermore, if the mixed equilibrium exists, $x_A + y_A = x_B + y_B$. 
\end{theorem}

A direct consequence of Theorem~\ref{th:mixed} is that coin--toss equilibria can only exist if the number of $A$--partisan--voters is smaller than the total number of $B$--supporters, i.e. $Npp_a \leq N(1-p_A)$. Once this condition is satisfied, the width of the interval of existence of coin--toss equilibria increases when $p$ and $p_A$ decrease. In Figure~\ref{fig:001} we sketch the bounds of the interval as a function of the population size $N$ for $p = 0.01$, i.e., for a low share of partisan--voters. Figure~\ref{fig:02} shows how these bounds behave for $p = 0.2$, i.e., when then share of partisan--voters is high. Figure~\ref{fig:limits} depicts some situations close to the limit $Npp_a = N(1-p_A)$ for different values of $p$.

\begin{figure}
    \centering
    \subfloat[$p_A = 0.52$]{\includegraphics[scale=0.5]{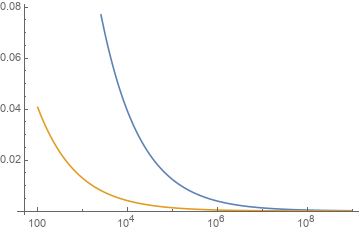}}
    \hfill
    \subfloat[$p_A = 0.75$]{\includegraphics[scale=0.5]{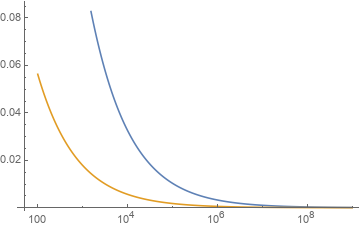}}
    \hfill
    \subfloat[$p_A = 0.90$]{\includegraphics[scale=0.5]{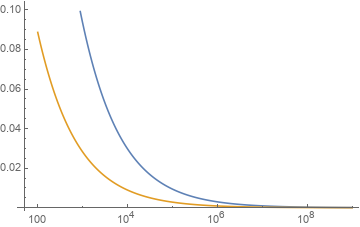}}
    \hfill
    \subfloat[$p_A = 0.99$]{\includegraphics[scale=0.5]{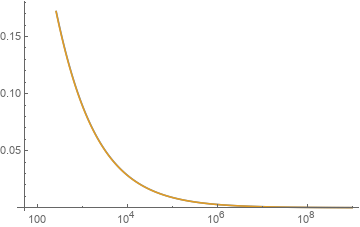}}  
    \caption{Bounds of the interval of existence of coin--toss equilibria for $p = 0.01$ as a function of the population size $N$.}
    \label{fig:001}
\end{figure}

\begin{figure}
    \centering
    \subfloat[$p_A = 0.52$]{\includegraphics[scale=0.5]{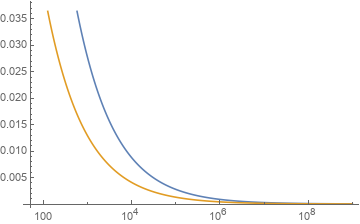}}
    \hfill
    \subfloat[$p_A = 0.60$]{\includegraphics[scale=0.5]{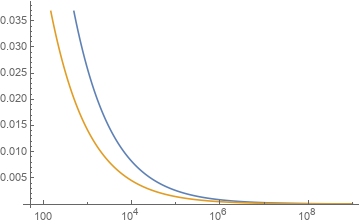}}
    \hfill
    \subfloat[$p_A = 0.75$]{\includegraphics[scale=0.5]{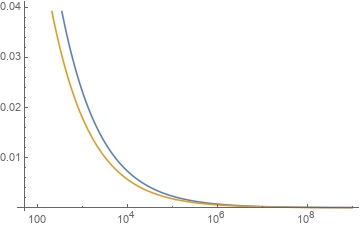}}
    \hfill
    \subfloat[$p_A = 0.83$]{\includegraphics[scale=0.5]{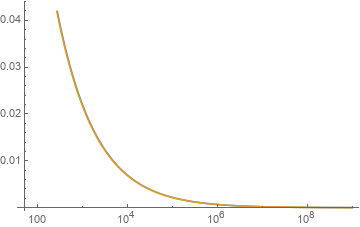}}
    \caption{Bounds of the interval of existence of coin--toss equilibria for $p = 0.2$ as a function of the population size $N$.}
    \label{fig:02}
\end{figure}

\begin{figure}
    \centering
    \subfloat[$p = 0.2, p_A = 0.83$]{\includegraphics[scale=0.5]{p20pA83.png}}
    \hfill
    \subfloat[$p = 0.5, p_A = 0.66$]{\includegraphics[scale=0.5]{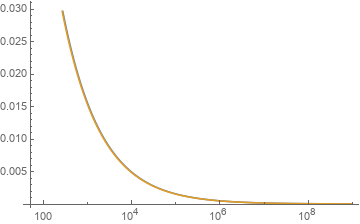}}
    \hfill
    \subfloat[$p = 0.75, p_A = 0.57$]{\includegraphics[scale=0.5]{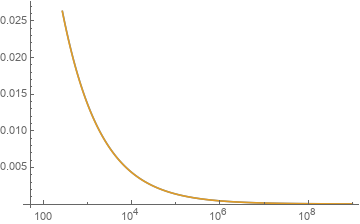}}
    \hfill
    \subfloat[$p = 0.9, p_A = 0.52$]{\includegraphics[scale=0.5]{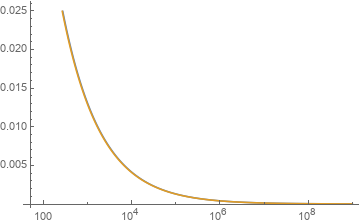}}
    \caption{Limit cases for different values of $p$.}
    \label{fig:limits}
\end{figure}

The last point of Theorem~\ref{th:mixed} implies, in short, that the parties A and B always receive an equal amount of votes in expectation in a coin--toss equilibrium. Hence, the outcome of the election does not agree with the public preference, and this is not desirable from a utilitarian welfare perspective. 

The second point of the theorem yields that a coin--toss equilibrium exists if the voting costs $c$ are in a particular interval, whose bounds decrease inversely proportionally to $\sqrt{N}$ for fixed $p$. More precisely, the upper bound decreases inversely proportionally to $\sqrt{x_A}$---where $x_A$ is the total amount of partisan voters of the larger party $A$. The lower bound decreases inversely proportionally to $\sqrt{N(1-p_A)}$---where $N(1-p_A)$ is the total amount of adherents of the smaller party $B$. Due to the square--root, this interval might be non--negligible, even for large populations. To avoid this situation, the voting costs should be either very low or they should be at least moderately high.

Next, we analyze the  existence of non--coin--toss equilibria. We will see that, unlike coin--toss equilibria, these always lead to the victory of the stronger party.

\subsection{A first set of border equilibria}

With border equilibria, we denote equilibria in which one or both values of $\al_A$ and $\al_B$ attain boundary values. In this subsection, we outline two types of such border equilibria. We start with the following case. 

\begin{pr}\label{pr:partmix}
Let $(\al_A,\al_B)$ be a partial absenteeism equilibrium. Then, $\al_A = 0$. Furthermore, if $(0,\al_B)$ is not a coin--toss equilibrium, $x_A > x_B + y_B$.
\end{pr}

\begin{proof}
In a partial absenteeism equilibrium, we have either $\al_A = 0$ or $\al_B = 0$. 

First, assume that $y_B=\al_B=0$. In that case, we need to satisfy (\ref{eq:1G}) with equality and $c\geq R_2(y_A,0)$. We calculate
\begin{equation}
	R_2(y_A,0)-R_1(y_A,0)
	 =\tfrac12(x_A-x_B+y_A)\0F1(;2,x_B(x_A+y_A))e^{-x_A-x_B-y_A}>0,\label{eq:3B}
\end{equation}
since $x_A>x_B$, $y_A>0$, and the argument of the generalized hypergeometric function $\0F1$ in
(\ref{eq:3B}) is positive. From (\ref{eq:1G}) we obtain $R_2(y_A,0)-c>0$, a contradiction. 

Hence,
no partial absenteeism equilibrium with $\al_B=0$ and $\al_A \neq 0$ exists and we have proved the first point of the theorem.

Now, suppose $(0, \al_B)$ is a partial absenteeism equilibrium, which is not a coin--toss equilibrium. Then, (\ref{eq:2G}) must
be satisfied with equality, and $c> R_1(0,y_B)$. Since
\begin{equation}
\begin{aligned}
	R_2(0,y_B)-R_1(0,y_B)
	 &=\tfrac12(x_A-x_B-y_B)\0F1(;2,x_A(x_B+y_B))e^{-x_A-x_B-y_B} \\
	&= c - R_1(0,y_B) > 0. 
\end{aligned}
\end{equation}
Therefore, $x_A-x_B-y_B > 0$.
\end{proof}

Proposition~\ref{pr:partmix} is essential, since it ensures that a partial absenteeism equilibrium, which is not coin--toss, leads to party A's victory. A similar result holds for partial saturation equilibria.

\begin{pr}\label{pr:partmixsat}
Let $(\al_A,\al_B)$ be a partial saturation equilibrium. Then, $\al_B = 1$. Furthermore, if $(\al_A,1)$ is not coin--toss, $x_A + y_A > N(1-p_A) = x_B + y_B$.
\end{pr}

We have not yet studied the existence of non coin--toss equilibria, but we know already that these do not go against public opinion if they occur.

Hence, to sum up, so far we have shown that all equilibria, which are not coin--toss, i.e. the partial equilibria, the no--queue equilibrium and the all--swipe equilibrium, lead to the victory of the stronger party $A$. 

\subsection{Partial absenteeism equilibria with $y_A=\al_A=0$ and the no--queue equilibrium}

Next, we move on to a quantitative---and very technical study---of the existence and uniqueness of partial absenteeism equilibria with $y_A=\al_A=0$. We obtain

\begin{pr}\label{pr:partmixedex}
Define for all $z, x_A \in \Rz_+$, 
$$
	h(x_A,z):=\tfrac12\bigl(\0F1(;1,x_Az)+x_A\,\0F1'(;1,x_Az)\bigr)e^{-x_A-z}.
$$
If $x_A > \sqrt{2}$, then:
\begin{enumerate}
\item
For $\min(h(x_A,x_B), h(x_A,x_A)) \leq c <\max(h(x_A,x_B), h(x_A,x_A))$, one partial absenteeism equilibrium exists.

\item
If $0 \leq c <\min(h(x_A,x_B), h(x_A,x_A))$ no partial absenteeism equilibrium exists. 
\end{enumerate}
If $x_A \leq \sqrt{2}$ and $h(x_A,x_A)\leq c\leq h(x_A, x_B)$, one partial absenteeism equilibrium exists.

\end{pr}

The distinction of cases in the statement of Proposition~\ref{pr:partmixedex} is far from being exhaustive. At the end of the proof, the reader can find a detailed description of all possible cases. However, the conditions of the cases, which are not mentioned in the statement of Proposition~\ref{pr:partmixedex}, are difficult to verify, since they involve the unique maximum $x_A^*$ of $h(x_A, z)$, which cannot be computed easily. We will also see that they are less relevant for the conclusion of this work.

Now, in general, we cannot simplify the distinction of cases in the proof of Proposition~\ref{pr:partmixedex}. However, we will see that $x_A^*$ is not much smaller than $x_A$. In particular, for fixed $p$, there exists $q < 1$, such that $x_B = q x_A$---namely $q = \frac{p_B}{p_A}$---and, $x_B$ is always smaller than $x_A^*$ for large enough $N$. The key observation is the following Lemma.

\begin{lem}\label{le:partmixedasymp}
Let $0 < q < 1$. Then, there exists $x_0 \in \Rz_+$ such that $h(x_A,qx_A)$ is decreasing in $x_A$ on $[x_0, + \infty[$ and
\begin{equation}\label{eq:ratqest}
h(x_A, qx_A) =_{x_A \to + \infty} \frac{\sqrt{q} + 1}{4\sqrt{\pi x_A}q^{\frac{3}{4}}}e^{-(q+1 - 2\sqrt{q})x_A} + O\biggl(\frac{e^{-(q+1 - 2\sqrt{q})x_A}}{x_A^{\frac{3}{2}}}\biggr).
\end{equation}
Note that $q+1 - 2\sqrt{q} > 0$. Furthermore, $h(qx_A, x_A) \leq h(x_A, qx_A)$, such that (\ref{eq:ratqest}) also gives an asymptotic estimate of $h(qx_A, x_A)$. Finally, for $q>1$, 
\begin{equation}
h(x_A,qx_A) =_{x_A \to + \infty} \frac{\cosh(\tfrac14\log(q))}{2\sqrt{\pi x_A}}e^{-(q+1 - 2\sqrt{q})x_A} + O\biggl(\frac{e^{-(q+1 - 2\sqrt{q})x_A}}{x_A^{\frac{3}{2}}}\biggr).
\end{equation}
\end{lem}

We deduce our main result on partial absenteeism equilibria.
\begin{theorem}\label{th:partmixed}
As previously defined, for all $z, x_A \in \Rz_+$, 
$$
	h(x_A,z):=\tfrac12\bigl(\0F1(;1,x_Az)+x_A\,\0F1'(;1,x_Az)\bigr)e^{-x_A-z}.
$$
For fixed $p$ and $p_A$ and $N$ large enough, $h(Npp_A, Np(1-p_A)) = h(x_A, x_B) < h(x_A, x_A) = h(Npp_A, Npp_A)$ and there exists:
\begin{enumerate}
\item exactly one partial absenteeism equilibrium if $$h(Npp_A, Np(1-p_A)) \leq c < h(Npp_A, Npp_A) = \frac{1}{2} g(2x_A),$$
\item no partial absenteeism equilibrium if $c < h(Npp_A, Np(1-p_A))$.
\end{enumerate}
Furthermore, the upper bound is decreasing, and there exists $N_0$ such that for all $N>N_0$, the lower bound is decreasing on $[N_0, +\infty[$. Asymptotically, the lower and upper bounds behave like 
\begin{equation}
\begin{aligned}
h(Npp_A, Np(1-p_A)) &=_{N \to +\infty} O\bigl(e^{-p\bigl(1 - 2\sqrt{p_A(1-p_A)}\bigr)N}\bigr) \\
h(Npp_A, Npp_A) &=_{N \to +\infty} \sqrt{\frac{1}{4\pi pp_A N}} + O\biggl(\frac{1}{N^{\frac{3}{2}}}\biggr).
\end{aligned}
\end{equation}
\end{theorem}
\begin{proof}[Proof of Theorem~\ref{th:partmixed}]
We first prove the asymptotic behavior of $h(Npp_A, Np(1-p_A))$ and $h(Npp_A, Npp_A)$ for $N \to +\infty$. Note that $2h(Npp_A, Npp_A) = g(2Npp_A)$. Hence, the second point follows directly from (\ref{eq:gasymp}). The first point, on the other hand, is a direct consequence of Lemma~\ref{le:partmixedasymp}. We have already proved in (\ref{eq:usediff}) that the upper bound is decreasing in $N$. The result on the lower bound follows once again from Lemma~\ref{le:partmixedasymp}.

Now, it follows that $h(Npp_A, Np(1-p_A)) < h(Npp_A, Npp_A)$ for all $N$ larger than a certain $N_0$, since $h(Npp_A, Npp_A)$ behaves like $\frac{1}{N}$ and $h(Npp_A, Np(1-p_A))$ is bounded by an exponentially decreasing function. 

Since $z \mapsto h(x_A,z)$ has a unique global maximizer $x_A^* < x_A$, it follows for $N \geq N_0$ that $x_B < x_A^* <x_A$. Now, Proposition~\ref{pr:partmixedex} implies the part of the statement concerning the existence of equilibria.
\end{proof}

Using Theorem~\ref{th:partmixed}, studying the existence of the no--queue equilibrium is straightforward. Indeed, the lower bound of its domain of existence coincides with the lower bound of the domain of existence of partial absenteeism equilibria. However, that domain has no upper bound. We obtain the following result.
\begin{co}\label{th:noshow}
The no--queue equilibrium $(0,0)$ exists if $c \geq h(x_A, x_B)$. For fixed $p, p_A$, the asymptotic behaviour of that bound is given by  
\begin{equation}
h(Npp_A, Np(1-p_A)) =_{N \to +\infty} O\bigl(e^{-p\bigl(1 - 2\sqrt{p_A(1-p_A)}\bigr)N}\bigr).
\end{equation}
\end{co}

\begin{proof}

For the existence of the no--queue equilibrium, the inequalities
\begin{align}
	R_1(0,0) &= \tfrac12 e^{-x_A-x_B}(\0F1(;1,x_Ax_B)+x_B\;\0F1(;2,x_Ax_B)) \leq c \label{eq:1H}\\
	R_2(0,0) &= \tfrac12 e^{-x_A-x_B}(\0F1(;1,x_Ax_B)+x_A\;\0F1(;2,x_Ax_B)) \leq c \label{eq:2H}
\end{align}
need to be satisfied. Since, $R_1(0,0) \leq R_2(0,0)$, this is equivalent to merely requiring (\ref{eq:2H}), i.e.,
\begin{equation}
h(x_A, x_B) \leq c.
\end{equation} 
The asymptotic behavior of $h(x_A, x_B)$ for $N \to +\infty$ has already been shown in Theorem~\ref{th:partmixed}.
\end{proof}

In particular, we deduce the existence of at least one equilibrium for $c > h(x_A,x_B)$, a threshold that decreases exponentially in $N$.

Therefore, the lower bound of the domain of costs $c$ on which partial absenteeism equilibria and the no--show equilibrium exist decreases exponentially in the population size $N$. On the other hand, the lower bound of the domain of costs on which the coin toss--equilibria exist only decreases like $\frac{1}{\sqrt{N}}$. We have shown that the latter is entirely contained in the former for sufficiently large $N$. 

\subsection{Partial saturation equilibria with $\al_B = 1$ and the all--swipe equilibrium}

Finally, we study the existence and uniqueness of partial saturation equilibria with $\al_B=1$, or equivalently $y_B + x_B = N(1-p_A)$. This is the last type of border equilibria. The analogon of Theorem~\ref{th:partmixed} for partial saturation equilibria is actually much simpler.

\begin{theorem}\label{th:partmixedexsat}
Define for all $z, x_A \in \Rz_+$, 
$$
	h(N(1-p_A),z):=\tfrac12\bigl(\0F1(;1,N(1-p_A)z)+N(1-p_A)\,\0F1'(;1,N(1-p_A)z)\bigr)e^{-N(1-p_A)-z}.
$$
Then:
\begin{enumerate}

\item
For $h(N(1-p_A), Np_A) \leq c \leq h(N(1-p_A),N(1-p_A)) =  \frac{1}{2}g(2N(1 - p_A))$ one partial saturation equilibrium exists.

\item
Otherwise, no partial saturation equilibrium exists.
\end{enumerate}

The upper bound is decreasing and asymptotically behaves like
\begin{equation}
\begin{aligned}
\frac{1}{2}g(2N(1 - p_A)) =_{N \to +\infty} \sqrt{\frac{1}{4\pi N(1-p_A)}} + O\biggl(\frac{1}{N^{\frac{3}{2}}}\biggr).
\end{aligned}
\end{equation}

Furthermore, there exists $N_0$ such that for all $N>N_0$, the lower bound is decreasing on $[N_0, +\infty[$. Asymptotically, it is dominated by  
\begin{equation}
h(N(1-p_A), Np_A) =_{N \to +\infty} O\bigl(e^{-\bigl(1 - 2\sqrt{p_A(1-p_A)}\bigr)N}\bigr)
\end{equation}

\end{theorem}

Using Theorem~\ref{th:partmixedexsat}, studying the existence of the last type of equilibria, i.e., the all--swipe equilibrium, is straightforward. Indeed, the upper bound of its domain of existence coincides with the lower bound of the domain of existence of partial saturation equilibria. The lower bound of the domain is of course $0$. We obtain the following result.
\begin{co}\label{th:allswipe}
The all--swipe equilibrium $(1,1)$ exists if $c \leq h(N(1-p_A), Np_A)$. 
\end{co}

\subsection{Summary}

 Let us summarize the most important results we have obtained in our study of the existence and uniqueness of equilibria depending on the voting costs $c$.

\begin{theorem}

We fix $0<p<1$ and $\frac{1}{2} < p_A < 1$ and denote the expected population size by $N$. The expected number of $A$--partisan voters is $x_A = Npp_A$ and the expected number of $B$--partisan voters is $x_B = Np(1 - p_A)$. The total numbers of $A$ and $B$--supporters are $Np_A$ and $N(1-p_A)$ respectively.

Recall that we have defined for all $z,x_A \in \Rz_+$,
\begin{equation}
g(z):=(I_0(z)+I_1(z))e^{-z}
\end{equation}
and
\begin{equation}
h(x_A,z):=\tfrac12\bigl(\0F1(;1,x_Az)+x_A\,\0F1'(;1,x_Az)\bigr)e^{-x_A-z}.
\end{equation}

We have seen in Lemma~\ref{le:partmixedasymp} that there exists $N_0$ such that for $N > N_0$,
\begin{enumerate}
    \item $\frac{1}{2}g(2Npp_A) = h(Npp_A, Npp_A) > h(Npp_A, Np(1 - p_A))$.
    \item $h(Npp_A,Np(1-p_A)) > h(Np(1-p_A), Npp_A) > h(N(1-p_A), Np_A)$.
\end{enumerate}
Hence, in the most relevant case, i.e. $x_A > \sqrt{2}$, $N$ sufficiently large and $Npp_A \leq N(1-p_A)$, the following situations can occur depending on the value of $c$:
\begin{enumerate}
\item $c > \frac{1}{2}g(2Npp_A)$. There exist either $2$, $1$ or no partial absenteeism equilibria and the no--queue equilibrium. All of these lead to the victory of party $A$.

\item $ \frac{1}{2}g(2Npp_A) \geq c \geq \frac{1}{2}g(2N(1 - p_A))$. There exists one coin-toss equilibrium, one partial absenteeism equilibrium and the no--queue equilibrium. In the coin-toss equilibrium, party $A$ and party $B$ obtain the same number of votes in expectation.

\item $h(Npp_A,Np(1-p_A)) \leq c < \frac{1}{2}g(2N(1 - p_A))$. There exist one partial absenteeism equilibrium, one partial saturation equilibrium and the no--queue equilibrium. All lead to the victory of party $A$.

\item $h(N(1-p_A), Np_A) \leq c \leq h(Npp_A,Np(1-p_A))$. There exists exactly one partial saturation equilibrium.

\item $c < h(N(1-p_A), Np_A)$. Then, only the all--swipe equilibrium exists.
\end{enumerate}

Thus, we have 4 decreasing frontiers, which do not intersect for $N$ sufficiently large. While $\frac{1}{2}g(2Npp_A)$ and $\frac{1}{2}g(2N(1 - p_A))$ decrease like $\frac{1}{\sqrt{N}}$ in $N$, $h(Npp_A,Np(1-p_A))$ and $h(N(1-p_A), Np_A)$ decrease exponentially in $N$. 
\end{theorem}

The most important conclusion is as follows. 
The only domain we need to avoid actively is the domain of existence of coin-toss equilibria. The bounds of that domain decrease like $\frac{1}{\sqrt{N}}$ in $N$, which might be non--negligible, even for large populations. Avoiding that domain can either be achieved by setting $c$ at a moderately high level, or by setting very low voting costs $c$.

\section{Equivalent models}\label{sec:equivalent}

In this section, we show that the equilibria may be recovered in alternative versions of the game. First, we consider a game in which voters select only pure strategies, but we drop type symmetry. Second, we assume that two parties compete and can coordinate how many supporters are going to vote.

\subsection{The pure strategy game}

Let us define the static game $\mathcal{G}'$ in a similar way to $\mathcal{G}$, but this time with pure strategies. 

Suppose that $(\sigma_i)_{i \in [\![1, N(p-1)]\!]}$ is an equilibrium. Although the utility of player $i$ in the state $(\sigma_i)$ is the same as before, technically, we cannot reason analogously as in the case of  $\mathcal{G}'$, since $(\sigma_i)$ is not supposed to be type symmetric. Quantitatively, suppose the number of $A$--supporters who chose strategy $1$ is $y_A$ and that player $i$ is among them. Furthermore, denote the number of $B$--supporters who chose strategy $1$ by $y_B$. Then, the number of $A$--supporters other than $i$ who chose strategy $1$ is $y_A -1$. Hence, 
\begin{align}
    U_i(1,(\sigma_j)_{j \neq i}) \coloneqq  u_A\biggl(1, \frac{y_A - 1}{N(1-p)p_A - 1}, \frac{\sum_{j \in I_B}\sigma_j}{N(1-p)(1-p_A)}\biggr)
\end{align}
and
\begin{align}
    U_i(0,(\sigma_j)_{j \neq i}) \coloneqq u_A\biggl(0, \frac{y_A - 1}{N(1-p)p_A - 1}, \frac{\sum_{j \in I_B}\sigma_j}{N(1-p)(1-p_A)}\biggr).
\end{align}
On the other hand, for an $A$-supporter $j$, who chose strategy $0$, the number of $A$--supporters other than $j$ who chose strategy $1$, is $y_A$. Thus, the player's utilities are
\begin{align}
    U_i(1,(\sigma_j)_{j \neq i}) \coloneqq  u_A\biggl(1, \frac{y_A}{N(1-p)p_A - 1}, \frac{\sum_{j \in I_B}\sigma_j}{N(1-p)(1-p_A)}\biggr)
\end{align}
and
\begin{align}
    U_i(0,(\sigma_j)_{j \neq i}) \coloneqq u_A\biggl(0, \frac{y_A}{N(1-p)p_A - 1}, \frac{\sum_{j \in I_B}\sigma_j}{N(1-p)(1-p_A)}\biggr).
\end{align}
Hence, unlike in the type--symmetric states of game $\mathcal{G}$, the utilities of the different $A$--supporters in state $(\sigma_i)$ are not all identical---except if all chose strategy $1$ or all chose strategy $0$. Thus, we cannot reason as in the case of $\mathcal{G}$.

However, the utility function is defined in a way that these differences in the utilities of the different players are small and tend to $0$ if the number of players $N$ becomes very large. Indeed, for large $N$, $\frac{y_A}{N(1-p)p_A - 1}$ and $\frac{y_A - 1}{N(1-p)p_A - 1}$ are approximately equal. Thus, for large $N$, our previous mathematical discussion on type--symmetric equilibria of $\mathcal{G}$ can more or less be adopted one--to--one to describe all possible equilibria of $\mathcal{G'}$. Coin-toss equilibria, partial equilibria, the no--queue and the all--swipe equilibria exist for similar domains of voting costs and the corresponding $(\alpha_A,\alpha_B)$ are almost similar to those we have obtained for the equilibria of $\mathcal{G'}$.

However, the equilibria are, of course, not unique anymore, as permuting the players has no effect on the utilities.

\subsection{The parties decide}

Surprisingly, another game $\mathcal{H}$, whose equilibria are precisely the type-symmetric equilibria of $\mathcal{G}$, is the following: we model electoral behavior as a two--player game $\mathcal{H}$ from the point of view of the parties $A$ and $B$. In this alternative model, the parties decide how many of their supporters vote by imposing the pair of probabilities $(\alpha_A, \alpha_B)$. Formally, the set of possible strategies is $[0,1]^2$ and the utility functions are given by 
\begin{align}
	U_A(\alpha_A, \alpha_B) &= \mathbb{E}[f(x_A + y_A, x_B + y_B) - c y_A], \\
        U_B(\alpha_A, \alpha_B) &= \mathbb{E}[f(x_B + y_B, x_A + y_A) - c y_B].
\end{align}
Note that the equilibrium condition of an actual Nash equilibrium without constraints is 
\begin{align}
	\frac{d}{d \alpha_A} U_A(\alpha_A, \alpha_B) &= \frac{d}{d \alpha_A}\mathbb{E}[f(x_A + y_A, x_B + y_B) - c y_A] = 0, \\
        \frac{d}{d \alpha_B} U_B(\alpha_A, \alpha_B) &= \frac{d}{d \alpha_B}\mathbb{E}[f(x_B + y_B, x_A + y_A) - c y_B] = 0.
\end{align}
A short computation shows that this is equivalent to
\begin{align}\label{eq:HNash}
  &c = \sum_{a=0}^{\infty}\sum_{b=0}^{\infty}P_{x_A}(a)P_{x_B}(b)\sum_{r=0}^{\infty}\sum_{s=0}^{\infty}P_{y_A}(r)P_{y_B}(s)(f(a+r+1, b+s) - f(a+r, b+s)), \\
  &c = \sum_{a=0}^{\infty}\sum_{b=0}^{\infty}P_{x_A}(a)P_{x_B}(b)\sum_{r=0}^{\infty}\sum_{s=0}^{\infty}P_{y_A}(r)P_{y_B}(s)(f(b + s + 1, a + r) - f(b+s, a+r)).  
\end{align}
Note that these are precisely the indifference conditions of type--symmetric equilibria $\mathcal{G}$ given by (\ref{A_supporter}) and (\ref{B_supporter}). Hence, although modeling electoral behavior by $\mathcal{H}$ is conceptually less reasonable than by $\mathcal{G}$---as parties cannot impose on anyone to vote in their favor--- the two games are mathematically equivalent. 


Finally, in the setting of $\mathcal{H}$, adding the constraints, the utility could be maximal at the border of $[0,1]^2$. 


The border conditions look like
\begin{align}
	\frac{d}{d \alpha_A} U_A(\alpha_A, 1) &= \frac{d}{d \alpha_A}\mathbb{E}[f(x_A + y_A, x_B + y_B) - c y_A] = 0, \\
        \frac{d}{d \alpha_B} U_B(\alpha_A, 1) &= \frac{d}{d \alpha_B}\mathbb{E}[f(x_B + y_B, x_A + y_A) - c y_B] > 0.
\end{align}
These are precisely those of partial equilibria in game $\mathcal{G}$.


\section{Conclusion}\label{sec:conclusion}

With this paper we have initiated the study of the optimal cost of voting as a variable for the design of democratic elections. Although the model is   simple, the analysis of equilibria turned out to be quite difficult. Yet, with appropriate tools, for all cost levels, all equilibria could be characterized. Somewhat surprisingly, intermediate costs of voting are undesirable, as they can allow the so-called ``coin-toss equilibria'' in which both proposals have the same chance of winning and thus the majority in the polity may not prevail. 

The model and the tools we present in this paper open a variety of extensions. For instance, the techniques can be used to examine equilibria in variations of the voting game, e.g. when voting costs may be different between non-partisan voters or when voting when further motives are added to the decision whether to vote or not as developed in \citep[]{LevineMattozzi2020,CoateBesley1998}. Moreover, estimating the actual voting costs and comparing them to the equilibrium thresholds derived in this paper may open an entire branch of applied work regarding the design of voting procedures.

\newpage

\bibliographystyle{apalike}
\bibliography{sample}

\begin{thebibliography}{}

\bibitem[Arzumanyan and Polborn, 2017]{polborn}
Arzumanyan, M. and Polborn, M. (2017).
\newblock Costly voting with multiple candidates under plurality rule.
\newblock {\em Games and Economic Behavior}, 106:38--50.

\bibitem[B\"{o}rgers, 2004]{borgers}
B\"{o}rgers, T. (2004).
\newblock Costly voting.
\newblock {\em American Economic Review}, 94(1):57--66.

\bibitem[Coate and Besley, 1998]{CoateBesley1998}
Coate, S. and Besley, T. (1998).
\newblock Sources of inefficiency in a representative democracy: {A} dynamic analysis.
\newblock {\em American Economic Review}, 8(1):139--156.

\bibitem[Feddersen and Sandroni, 2006a]{feddersen2006ethical}
Feddersen, T. and Sandroni, A. (2006a).
\newblock Ethical voters and costly information acquisition.
\newblock {\em Quarterly Journal of Political Science}, 1(3):287--312.

\bibitem[Feddersen and Sandroni, 2006b]{feddersen}
Feddersen, T. and Sandroni, A. (2006b).
\newblock A theory of participation in elections.
\newblock {\em American Economic Review}, 96(4):1271--1282.

\bibitem[Gersbach et~al., 2021]{AV}
Gersbach, H., Mamageishvili, A., and Tejada, O. (2021).
\newblock The effect of handicaps on turnout for large electorates with an application to {A}ssessment {V}oting.
\newblock {\em Journal of Economic Theory}, 195:105228.

\bibitem[Herrera et~al., 2014]{herrera_stuff}
Herrera, H., Morelli, M., and Palfrey, T. (2014).
\newblock {Turnout and power sharing}.
\newblock {\em Economic Journal}, 124:131--162.

\bibitem[Ledyard, 1984]{ledyard1984pure}
Ledyard, J.~O. (1984).
\newblock The pure theory of large two-candidate elections.
\newblock {\em Public Choice}, 44(1):7--41.

\bibitem[Levine and Mattozzi, 2020]{LevineMattozzi2020}
Levine, D.~K. and Mattozzi, A. (2020).
\newblock Voter turnout with peer punishment.
\newblock {\em American Economic Review}, 110(10):3298--3314.

\bibitem[Myerson, 1998]{poisson}
Myerson, R. (1998).
\newblock {Population uncertainty and Poisson games}.
\newblock {\em International Journal of Game Theory}, 27:375--392.

\bibitem[Palfrey and Rosenthal, 1983]{palfrey1983strategic}
Palfrey, T.~R. and Rosenthal, H. (1983).
\newblock A strategic calculus of voting.
\newblock {\em Public Choice}, 41(1):7--53.

\bibitem[Riker and Ordeshook, 1968]{RikerOrdeshook1968}
Riker, W. and Ordeshook, P. (1968).
\newblock A theory of the calculus of voting.
\newblock {\em American Political Science Review}, 62(1):25--42.

\bibitem[Watson, 1922]{watson1922treatise}
Watson, G.~N. (1922).
\newblock {\em A treatise on the theory of Bessel functions}, volume~3.
\newblock The University Press.

\bibitem[Xefteris, 2019]{xefteris}
Xefteris, D. (2019).
\newblock Strategic voting when participation is costly.
\newblock {\em Games and Economic Behavior}, 116:122--127.

\end{thebibliography}

\section{Appendix}

\subsection{Proofs of the main results}

\begin{proof}[Proof of Proposition~\ref{pr:rhsexpr}]
We first note that (\ref{B_supporter}) is 
obtained from (\ref{A_supporter}) by exchanging $x_A$ for $x_B$ and $y_A$ for $y_B$. Hence, we only need to 
simplify the r.h.s.\ of (\ref{A_supporter}). 

Before we start to calculate, we note that $f(m,n)=\tfrac12(1+\sgn(m-n))$. Therefore,
we have the r.h.s.\ of (\ref{A_supporter}):
\begin{equation}\label{eq:1A}
	R_1(y_A,y_B)=\frac12\sum_{a=0}^\infty\sum_{b=0}^\infty\sum_{r=0}^\infty\sum_{s=0}^\infty
		\frac{x_A^a}{e^{x_A}a!}\frac{x_B^b}{e^{x_B}b!}\frac{y_A^r}{e^{y_A}r!}
		\frac{y_B^s}{e^{y_B}s!}\bigl(\sgn(a+r+1-b-s)-\sgn(a+r-b-s)\bigr)
\end{equation}
Since $\sgn(n)-\sgn(n-1)=1$ for $n\in\Nz$ if and only if $n=0$ or $n=1$, we can remove the sum over
$s$, since the innermost sum vanishes unless $s=a+r+1-b$ or $s=a+r-b$. Furthermore, since $s\geq 0$,
in both cases we get $b\leq a+r+1$ and $b\leq a+r$, respectively. Hence (\ref{eq:1A}) becomes
\begin{equation}\label{eq:1B}
	R_1(y_A,y_B)=\frac1E\sum_{a=0}^\infty\frac{x_A^a}{a!}\sum_{r=0}^\infty\frac{y_A^r}{r!}\left(\sum_{b=0}^{a+r}
		\frac{x_B^b}{b!}\frac{y_B^{a-b+r}}{(a-b+r)!}+\sum_{b=0}^{a+r+1}
		\frac{x_B^b}{b!}\frac{y_B^{a-b+r+1}}{(a-b+r+1)!}\right),
\end{equation}
where $E=2e^{x_A+x_B+y_A+y_B}$. We can rearrange (\ref{eq:1B})
\begin{align}
	R_1(y_A,y_B)&=\frac1E\sum_{a=0}^\infty\frac{x_A^a}{a!}\sum_{r=0}^\infty\frac{y_A^r}{r!}\left(\frac1{(a+r)!}\sum_{b=0}^{a+r}
		\binom{a+r}b x_B^b y_B^{a-b+r}\right.\notag\\
		&\qquad\qquad+\left.\frac1{(a+r+1)!}\sum_{b=0}^{a+r+1}\binom{a+r+1}b
		x_B^b y_B^{a-b+r+1}\right),\notag\\
		&=\frac1E\sum_{a=0}^\infty\frac{x_A^a}{a!}\sum_{r=0}^\infty\frac{y_A^r}{r!}
			\left(\frac{(x_B+y_B)^{a+r}}{(a+r)!}+\frac{(x_B+y_B)^{a+r+1}}{(a+r+1)!}\right)\notag\\
		&=\frac1E\left(\sum_{a=0}^\infty\sum_{r=0}^\infty\frac{x_A^a}{a!}\frac{y_A^r}{r!}\frac{(x_B+y_B)^{a+r}}{(a+r)!}\right.\notag\\
		&\qquad\qquad\qquad+\left.\sum_{a=0}^\infty\sum_{r=0}^\infty\frac{x_A^a}{a!}\frac{y_A^r}{r!}\frac{(x_B+y_B)^{a+r+1}}{(a+r+1)!}\right) \notag\\
		&=\frac1E\left(\sum_{k=0}^\infty\frac{(x_A+y_A)^{k}}{k!}\frac{(x_B+y_B)^{k}}{k!}+\sum_{k=0}^\infty\frac{(x_A+y_A)^{k}}{k!}\frac{(x_B+y_B)^{k+1}}{(k+1)!}\right)\notag\\
		&=\frac1E\left(\sum_{k=0}^\infty\frac{((x_A+y_A)(x_B+y_B))^{k}}{(k!)^2}+
			(x_B+y_B)\sum_{k=0}^\infty\frac{((x_A+y_A)(x_B+y_B))^{k}}{k!(k+1)!}\right)\label{eq:1C}.
\end{align}
Finally, Definition~\ref{de:hypgeom} yields
\begin{align}
\notag\\
		R_1(y_A,y_B) &=\tfrac12\bigl(\0F1(;1,(x_A+y_A)(x_B+y_B)) \\
		&\qquad\qquad+(x_B+y_B)\,\0F1(;2,(x_A+y_A)(x_B+y_B))\bigl)e^{-x_A-y_A-x_B-y_B}.
\end{align}

Next, we substitute $x_A=Npp_A$, $x_B=Np(1-p_A)$, $y_A=N(1-p)p_A\al_A$, and $y_B=N(1-p)(1-p_A)\al_B$
into (\ref{eq:1D}) and obtain
\begin{align}
	r_1(\al_A,\al_B)&=\tfrac12\bigl(\0F1\bigl(;1,N^2p_A(1-p_A)(p+\al_A(1-p))(p+\al_B(1-p))\bigr)\notag\\
		&\qquad+N(1-p_A)(p+(1-p)\al_B)\notag\\
		&\qquad\qquad\cdot\0F1\bigl(;2,N^2p_A(1-p_A)(p+\al_A(1-p))(p+\al_B(1-p))\bigr)\bigl)\notag\\
		&\qquad\qquad\qquad\qquad\qquad\cdot e^{-N(p+(1-p)(p_A\al_A+(1-p_A)\al_B))}\label{eq:1E}
\end{align}
For equation (\ref{B_supporter}) we obtain, by exchanging $x_A$ with $x_B$ and $y_A$ with $y_B$, that the r.h.s. satisfies 
\begin{align}
	R_2(y_A,y_B)&=\tfrac12\bigl(\0F1(;1,(x_A+y_A)(x_B+y_B))\notag\\
		&\qquad\qquad+(x_A+y_A)\,\0F1(;2,(x_A+y_A)(x_B+y_B))\bigl)e^{-x_A-y_A-x_B-y_B}\\
	r_2(\al_A,\al_B)&=\tfrac12\bigl(\0F1\bigl(;1,N^2p_A(1-p_A)(p+\al_A(1-p))(p+\al_B(1-p))\bigr)\notag\\
		&\qquad+N p_A(p+(1-p)\al_A)\notag\\
		&\qquad\qquad\cdot\0F1\bigl(;2,N^2p_A(1-p_A)(p+\al_A(1-p))(p+\al_B(1-p))\bigr)\bigl)\notag\\
		&\qquad\qquad\qquad\qquad\qquad\cdot e^{-N(p+(1-p)(p_A\al_A+(1-p_A)\al_B))}\label{eq:2A}.
\end{align}
\end{proof}

\begin{proof}[Proof of Theorem~\ref{th:mixed}]
Recall that $(y_A,y_B)$ is a mixed equilibrium if and only if the conditions
\begin{align}
	c&=\tfrac12\bigl(\0F1(;1,(x_A+y_A)(x_B+y_B))\notag\\
		&\qquad\qquad+(x_B+y_B)\,\0F1(;2,(x_A+y_A)(x_B+y_B))\bigl)e^{-x_A-y_A-x_B-y_B}\label{eq:1G}\\
	c&=\tfrac12\bigl(\0F1(;1,(x_A+y_A)(x_B+y_B))\notag\\
		&\qquad\qquad+(x_A+y_A)\,\0F1(;2,(x_A+y_A)(x_B+y_B))\bigl)e^{-x_A-y_A-x_B-y_B}\label{eq:2G}
\end{align}
are satisfied. 
Subtracting the two equations and using the facts that $\exp$ is always positive and that
$\0F1$ is positive for non-negative arguments yields
$$
	0 = x_A-x_B+y_A-y_B.
$$
Hence, we have proved the last point of the lemma.
Inserting this into either one of the two equations yields
\begin{equation}\label{eq:3}
	c = \tfrac12\bigl(I_0(2(x_A+y_A))+I_1(2(x_A+y_A))\bigr)e^{-2(x_A+y_A)}.
\end{equation}
Define the function $g(z):=(I_0(z)+I_1(z))e^{-z}$. Corollary~\ref{co:Bessdiff} yields for all $z\in\Cz$,
\begin{equation}
\begin{aligned}
g(z) = (I_0(z)+I_0'(z))e^{-z}.
\end{aligned}
\end{equation}
Hence,
\begin{equation}\label{eq:usediff}
\begin{aligned}
zg(z)' &= z(I_0''(z) - I_0(z))e^{-z} \\
&= - I'_0(z) e^{-z},
\end{aligned}
\end{equation}
which is negative for all $z>0$, since $I_1(z)$ is a positive function on $\Rz_+$. Thus, $g$ is monotonically decreasing and satisfies $g(0)=1$. Next, we study the asymptotic behavior of $g$. Lemma~\ref{le:Bessasymp} immediately yields that 
\begin{equation}
I_0(z) + I_1(z) =_{+\infty} \frac{\sqrt{2} e^z}{\sqrt{\pi z}} + O\biggl(\frac{e^z}{z^\frac{3}{2}}\biggr)
\end{equation}
and, thus, 
\begin{equation}\label{eq:gasymp}
g(z) =_{+\infty} \sqrt{\frac{2}{\pi z}} + O\biggl(\frac{1}{z^\frac{3}{2}}\biggr).
\end{equation}
In particular, $\lim_{z\to\infty}g(z)=0$, therefore, for every $c\in(0,\tfrac12)$, there exists a unique $z_c$
satisfying $g(z_c)=2c$. 

However, not all of those $z_c$ are admissible equilibria. Since 
\begin{equation}
\begin{aligned}
x_A+y_A &= Np_A(p+(1-p)\al_A) \\
&= x_B+y_B = N(1 - p_A)(p+(1-p)\al_B),
\end{aligned}
\end{equation}
 $0<p,\al_A, \al_B<1$, and $\tfrac12<p_A<1$
there exists a solution $\al_A$ if and only if
\begin{equation}\label{eq:4}
	2Npp_A < z_c < 2N(1-p_A),
\end{equation}
i.e., the number of $A$--voters must always be greater than the number of $A$--partisan--voters and smaller than the total amount of $B$--supporters. Hence, 
\begin{equation}
	g(2x_A) > 2c > g(2N(1-p_A))
\end{equation}
needs to be satisfied.

In that case, there is a unique $\al^*_A$, and hence also
$$
	\al^*_B = \frac{p(2p_A-1)+p_A(1-p)\al_A^*}{(1-p)(1-p_A)}.
$$
So, if equation (\ref{eq:4}) is satisfied, there exists a unique mixed equilibrium.

The lower and upper bounds' asymptotics follow immediately from (\ref{eq:gasymp}).

\end{proof}

\begin{proof}[Proof of Proposition~\ref{pr:partmixsat}]
Let us suppose that a partial saturation equilibrium of the form
\begin{align}
	&R_1(N(1-p)p_A,y_B)\\ &\qquad= \tfrac12 e^{-Np_A-x_B-y_B}(\0F1(;1,Np_A(x_B+y_B))+(x_B+y_B)\;\0F1(;2,Np_A(x_B+y_B))) > c \\
	&R_2(N(1-p)p_A,y_B)\\ &\qquad= \tfrac12 e^{-Np_A-x_B-y_B}(\0F1(;1,Np_A(x_B+y_B))+Np_A\;\0F1(;2,Np_A(x_B+y_B))) = c ,
\end{align}
exists.
Then, subtracting the equations yields
\begin{align}
	R_2(N(1-p)p_A,y_B)&-R_1(N(1-p)p_A,y_B) \notag\\
	 &=\tfrac12(Np_A - (x_B + y_B))\0F1(;2,Np_A(x_B+y_B))e^{-Np_A-x_B-y_B} < 0.\label{eq:3Bbis}
\end{align}
But $Np_A > x_B + y_B$, and the argument of the generalized hypergeometric function $\0F1$ in
(\ref{eq:3Bbis}) is positive, so $R_2(N(1-p)p_A,y_B)-R_1(N(1-p)p_A,y_B)\geq 0$, a contradiction. 

Hence,
no partially mixed equilibrium with $\al_A=1$ exists and we have proved the first point of the theorem.

If $\al_B = 1$, (\ref{eq:1G}) must
be satisfied with equality, and $R_2(y_A,N(1-p)(1-p_A)) > c$ in the case of an equilibrium. Hence,
\begin{align}
	&R_2(y_A,N(1-p)(1-p_A))-R_1(y_A,N(1-p)(1-p_A)) \notag\\
	 &\qquad=\tfrac12(x_A+y_A-N(1-p_A))\0F1(;2,N(1-p_A)(x_A+y_A))e^{-x_A-y_A-N(1-p_A)} > 0\notag\\
\end{align}
and we conclude 
\begin{align}
	x_A+y_A-N(1-p_A) > 0.
\end{align}
\end{proof}

\begin{proof}[Proof of Proposition~\ref{pr:partmixedex}]
Again, (\ref{eq:2G}) must
be satisfied with equality, and $$c\geq R_1(0,y_B).$$ Since
\begin{align}
	R_2(0,y_B)&-R_1(0,y_B)
	 =\tfrac12(x_A-x_B-y_B)\0F1(;2,x_A(x_B+y_B))e^{-x_A-x_B-y_B}\notag\\
	 &=\tfrac12 N\bigl(p(2p_A-1)-(1-p)(1-p_A)\al_B\bigr)\notag\\
	 &\qquad\cdot\0F1\bigl(;2,N^2pp_A(1-p_A)(p+(1-p)\al_B)\bigr)e^{-N(p+(1-p)(1-p_A)\al_B)},\label{eq:3C}
\end{align}
in the case of an equilibrium,
\begin{align}
	x_A-x_B-y_B&\geq 0,\label{eq:3D}\\
	p(2p_A-1)-(1-p)(1-p_A)\al_B &\geq 0,\label{eq:3D2}
\end{align}
yielding
\begin{equation}\label{eq:3E}
	\al_B\leq \frac{p(2p_A-1)}{(1-p)(1-p_A)}.
\end{equation}
From (\ref{eq:2G}), we obtain
\begin{equation}\label{eq:partcond}
\begin{aligned}
	c &= \tfrac12\bigl(\0F1(;1,x_A(x_B+y_B))+x_A\,\0F1(;2,x_A(x_B+y_B))\bigr)e^{-x_A-x_B-y_B} \\
	&= h(x_A, x_B + y_B), 
\end{aligned}
\end{equation}
by Lemma~\ref{le:hgeomdiff}. Let us study $h(x_A, \ )$ in detail. We have $h(x_A,0)=\tfrac12(x_A+1)e^{-x_A}$. Furthermore, Lemma~\ref{le:hgeomdiff} yields  
\begin{equation}
h(x_A, z) = \tfrac12\bigl(\0F1(;1,x_Az)+x_A\,\0F1'(;1,x_Az)\bigr)e^{-x_A-z}.
\end{equation}
Thus, again by Lemma~\ref{le:hgeomdiff},
\begin{equation}
\begin{aligned}
z\frac{\d}{\d z} h(x_A, z) &= \tfrac12\bigl(-z\0F1(;1,x_Az)+x_A^2z\,\0F1''(;1,x_Az)\bigr)e^{-x_A-z} \\
&= \tfrac12\bigl((x_A - z)\0F1(;1,x_Az) - x_A\,\0F1'(;1,x_Az)\bigr)e^{-x_A-z}.
\end{aligned}
\end{equation}
Note that the additional factor $z$ does not change the sign of the above expression. Now, we verify readily that for all $z \in \Rz_+$, 
\begin{equation}\label{eq:auxfunc1}
i(z) := (x_A - z)\0F1(;1,x_Az) - x_A\,\0F1'(;1,x_Az) = \sum_{k\geq 1} \frac{- k(k^2 + k - x_A^2) }{x_A k! (k+1)!} (x_A z)^k.
\end{equation}
The polynomial $P(x_A,X) := -X(X^2 + X - x_A^2)$ has three---not necessarily distinct---roots $x_0 = 0$ and $x_{1,2} = \frac{-1 \pm \sqrt{1 + 4x_A^2}}{2}$. Furthermore, it is non--decreasing on $[0, \frac{-1 + \sqrt{1 + 3x_A^2}}{3}[$ and decreasing on $]\frac{-1 + \sqrt{1 + 3x_A^2}}{3}, +\infty[$. Therefore, $P(x_A,z)$ is in particular positive on $[0, \frac{-1 + \sqrt{1 + 4x_A^2}}{2}[$ and negative on $]\frac{-1 + \sqrt{1 + 4x_A^2}}{2}, +\infty[$.

For $x_A\leq\sqrt{2}$, all roots of $P(x_A,X)$ are smaller than or equal to $1$ and, hence, all coefficients of $i(z)$ are non--positive.  Hence, the function $h(x_A,z)$ monotonically decreases in $z$.

For $x_A>\sqrt{2}$, the above observations on $P(x_A,z)$ yield that there exists $N \in \Nz$ such that the coefficients $a_n := \frac{- k(k^2 + k - x_A^2) }{x_A k! (k+1)!}$ of the series in (\ref{eq:auxfunc1}) satisfy $a_n > 0$ for $1 \leq n \leq N+1$, $a_{N+2} \leq 0$ and $a_n < 0$ for $n > N+2$. Hence, the derivative of $i(z)$ satisfies the hypotheses of Lemma~\ref{le:incdec} and is positive in $0$. Thus, there exists $y_0 \in \Rz_+$ such that $i'(z)$ is positive for $z < y_0$ and negative for $z > y_0$. Hence, there exists $x_0 > y_0$ such that $i(z)$ is positive for $z < x_0$ and negative for $z > x_0$. Therefore, the function $h(x_A,z)$ has a unique global maximum $x_A^*$ on $\Rz_+$ for fixed $x_A$. We have $x_A\frac{\d}{\d z} h(x_A, x_A) < 0$, hence, $x_A^* < x_A$.

Now, we study the existence of partial absenteeism equilibria depending on $c$. As in the case of coin--toss equilibria, not all $z_c \in \Rz_+$ satisfying the equilibrium condition $h(x_A,z_c) = c$ are valid equilibria. The additional condition
\begin{equation}
x_B = Np(1-p_A) \leq z_c \leq Npp_A = x_A
\end{equation}
must be satisfied---where the second inequality follows from (\ref{eq:3D}).

Now, the following cases can occur if $x_A > \sqrt{2}$:

\begin{enumerate}

\item
If $c>h(x_A,x_A^*)$, then no partial absenteeism equilibrium exists.

\item 
For $c=h(x_A,x_A^*)$, there exists the unique equilibrium $(0,\al_B)$ with
$\al_B=\frac{x_A^*-Np(1-p_A)}{N(1-p)(1-p_A)}$ if $Np(1-p_A)\leq x_A^*$.

\item
For $\max(h(x_A,x_B), h(x_A,x_A)) \leq c<h(x_A,x_A^*)$ and $x_B < x_A^* < x_A$, 
there are two possible solutions, $z_1$ and $z_2$, of the equation $h(x_A,z)=c$, with
$x_B\leq z_1<x_A^*$ and $x_A \geq z_2>x_A^*$. These lead to the equilibria $(0,\al_{B,i})_{i \in {1,2}}$, with $\al_{B,1}=\frac{z_1-Np(1-p_A)}{N(1-p)(1-p_A)}$ and $\al_{B,2}=\frac{z_2-Np(1-p_A)}{N(1-p)(1-p_A)}$.

\item
For $\max(h(x_A,x_B), h(x_A,x_A)) \leq c<h(x_A,x_A^*)$ and $x_A^* < x_B < x_A$, the two above solutions are not admissible. Hence, no partial absenteeism equilibrium exists.

\item
For $\min(h(x_A,x_B), h(x_A,x_A)) \leq c <\max(h(x_A,x_B), h(x_A,x_A))$, only one of the above equilibria is admissible.

\item
If $0 \leq c <\min(h(x_A,x_B), h(x_A,x_A))$ no partial absenteeism equilibrium exists. 

\end{enumerate}
For $x_A\leq\sqrt{2}$, the cost $c$ needs to satisfy $2h(x_A,x_A)\leq 2c\leq 2h(x_A, x_B)$ for a partial absenteeism equilibrium of the form $(0,\al_B)$ to exist. Note that both these bounds are smaller than $2h(x_A,0) = (x_A+1)e^{-x_A}$.

\end{proof}

\begin{proof}[Proof of  Lemma~\ref{le:partmixedasymp}]
We have 
\begin{equation}
\begin{aligned}
h(x_A, qx_A) &= \tfrac12\bigl(\0F1(;1,qx_A^2)+x_A\,\0F1(;2,qx_A^2)\bigr)e^{-(q+1)x_A} \\
&= \tfrac12\bigl(I_0(2\sqrt{q}x_A) + \frac{1}{\sqrt{q}}I_1(2\sqrt{q}x_A)\bigr)e^{-(q+1)x_A}.
\end{aligned}
\end{equation}
Now, Lemma~\ref{le:Bessasymp} yields the second point. To show that there exists $x_0 \in \Rz_+$ such that $h(x_A,qx_A)$ is decreasing in $x_A$ on $[x_0, + \infty[$, we compute the derivative of $h(x_A,qx_A)$ with respect to $x_A$. Then, we proceed as in (\ref{eq:usediff}), using the first two terms of the asymptotic expansions in Lemma~\ref{le:Bessasymp} and verifying that there exists $x_0 \in \Rz_+$, such that $\frac{\d}{\d x_A}h(x_A,qx_A)$ is negative on $[x_0, + \infty[$.

For the third point we note that $q x_A\leq x_A$ and, hence, 
\begin{equation}
\begin{aligned}
h(x_A, qx_A) &= \tfrac12\bigl(\0F1(;1,qx_A^2)+x_A\,\0F1(;2,qx_A^2)\bigr)e^{-(q+1)x_A} \\
&\geq \tfrac12\bigl(\0F1(;1,qx_A^2)+qx_A\,\0F1(;2,qx_A^2)\bigr)e^{-(q+1)x_A} = h(qx_A,x_A).
\end{aligned}
\end{equation}

Finally, for $q>1$, $h(x_A,qx_A) = h(\tfrac1q (qx_A),qx_A)$, and the previous point gives 
\begin{equation}
h(\tfrac1q (qx_A),qx_A) =_{x_A \to + \infty} \frac{\cosh(\tfrac14\log(q))}{2\sqrt{\pi x_A}}e^{-(q+1 - 2\sqrt{q})x_A} + O\biggl(\frac{e^{-(q+1 - 2\sqrt{q})x_A}}{x_A^{\frac{3}{2}}}\biggr).
\end{equation}
\end{proof}

\begin{proof}[Proof of Theorem~\ref{th:partmixedexsat}]
This time, (\ref{eq:1G}) must
be satisfied with equality, and $R_2(y_A,N(1-p)(1-p_A)) > c$. Since
\begin{align}
	&R_2(y_A,N(1-p)(1-p_A))-R_1(y_A,N(1-p)(1-p_A)) \notag\\
	 &\qquad=\tfrac12(x_A+y_A-N(1-p_A))\0F1(;2,N(1-p_A)(x_A+y_A))e^{-x_A-y_A-N(1-p_A)}\notag\\
	 &\qquad=\tfrac12 N\bigl((p+1)p_A - 1 + (1-p)p_A\al_A \bigr)\notag\\
	 &\qquad \qquad\cdot\0F1\bigl(;2,N^2(1-p_A)(pp_A + (1-p)p_A\al_A))\bigr)e^{-N((p-1)p_A +(1-p)p_A\al_A +1)},\label{eq:3Cbis}
\end{align}
in the case of an equilibrium,
\begin{align}
	x_A+y_A-N(1-p_A)& > 0,\label{eq:3Dbis}\\
	(p+1)p_A - 1 + (1-p)p_A\al_A &> 0,\label{eq:3D2bis}
\end{align}
yielding
\begin{equation}\label{eq:3Ebis}
	\al_A > \frac{1 - (p+1)p_A}{(1-p)p_A}.
\end{equation}
From (\ref{eq:1G}), we obtain
\begin{equation}\label{eq:partcondbis}
\begin{aligned}
	c &= \tfrac12\bigl(\0F1(;1,N(1-p_A)(x_A+y_A))+N(1-p_A)\,\0F1(;2,N(1-p_A)(x_A+y_A))\bigr)e^{-x_A-y_A-N(1-p_A)} \\
	&= h(N(1-p_A), x_A + y_A), 
\end{aligned}
\end{equation}
by Lemma~\ref{le:hgeomdiff}. 

As in the proof of Proposition~\ref{pr:partmixedex} for $N(1-p_A)\leq\sqrt{2}$, the function $h(N(1-p_A),z)$ monotonically decreases in $z$.

Again, as in the proof of Proposition~\ref{pr:partmixedex}, for $N(1-p_A)>\sqrt{2}$, the function $h(N(1-p_A),z)$ is increasing from $0$ to its unique global maximum $z^*$ and decreasing afterwards, with $\lim_{z \to +\infty} h(N(1-p_A),z) = 0$ for fixed $N(1-p_A)$. Furthermore, we have, $z^* < N(1-p_A)$.

Now, we study the existence of partially mixed equilibria depending on $c$. As previously, not all $z_c \in \Rz_+$ satisfying (\ref{eq:partcondbis}) are valid equilibria. The additional condition
\begin{equation}
z^* < N(1-p_A) \leq z_c \leq Np_A
\end{equation}
must be satisfied---where the first inequality follows from (\ref{eq:3Dbis}).

Now, since $z^* < N(1-p_A)$ the following cases can occur:

\begin{enumerate}

\item
For $h(N(1-p_A), Np_A) \leq c \leq h(N(1-p_A),N(1-p_A))$ there exists exactly one solution $z_0$ of the equation $h(N(1-p_A),z) = c$ in the admissible interval $[N(1-p_A),Np_A]$. Hence, there exists a unique partial saturation equilibrium $(\al_A,1)$, with $\al_A = \frac{z_0 - Npp_A}{N(1-p)p_A}$.

\item
Otherwise, no partial saturation equilibrium exists.

\end{enumerate}
\end{proof}

\begin{proof}[Proof of  Corollary~\ref{th:allswipe}]

For the existence of the all--swipe equilibrium, the inequalities
\begin{align}
	&R_1(N(1-p)p_A,N(1-p)(1-p_A)) \\ &\qquad = \tfrac12 e^{-N}(\0F1(;1,N^2p_A(1-p_A))+N(1-p_A)\;\0F1(;2,N^2p_A(1-p_A))) \geq c \label{eq:1Hbis}\\
	&R_2(N(1-p)p_A,N(1-p)(1-p_A)) \\ &\qquad= \tfrac12 e^{-N}(\0F1(;1,N^2p_A(1-p_A))+Np_A\;\0F1(;2,N^2p_A(1-p_A))) \geq c \label{eq:2Hbis}
\end{align}
need to be satisfied. Since, $R_1(Np_A,N(1-p_A)) \leq R_2(Np_A,N(1-p_A))$, this is equivalent to merely requiring (\ref{eq:1Hbis}), i.e.,
\begin{equation}
h(N(1-p_A), Np_A) \geq c.
\end{equation} 
The asymptotic behavior of $h(N(1-p_A), Np_A)$ for $N \to +\infty$ has already been shown in Theorem~\ref{th:partmixed}.
\end{proof}

\subsection{Some additional useful results}

The following Lemma is used several times.

\begin{lem}\label{le:incdec}
Let $(a_n)_{n\in\Nz}$ be a sequence of real numbers such that there exists an integer $N \geq 1$ such that for all $n \leq N$, $a_n > 0$, $a_{N+1} \leq 0$ and for all $n > N + 1$, $a_n < 0$. Let 
\begin{equation}
\begin{aligned}
f: &\Rz_+ \to \Rz \\
&x \mapsto \sum_{n \geq 0} a_n x^n,
\end{aligned}
\end{equation}
and suppose that $f$ is well--defined on $\Rz$. Then there exists $x_0 \in \Rz$, such that $f$ is increasing on $[0, x_0[$ and decreasing on $]x_0, \infty[$. 
\end{lem}
\begin{proof}
We prove the result by induction on $N$. 

If $N = 1$, then $f''$ is a power series with all non--positive coefficients and negative coefficients of non--zero order, hence in particular negative. Thus $f'$ is decreasing and satisfies $f'(0) = a_1 > 0$. The other terms of order greater or equal to $2$ of $f'$ are all negative, so $\lim_{x \to \infty} f'(x) = -\infty$. Hence, $f$ is first increasing and then decreasing.

Now, suppose we have proved the result for some $N\in \Nz$ and let us prove the result for $N+1$. By induction hypothesis, there exists $x_0 \in \Rz$ such that $f'$ is increasing on $[0, x_0[$ and decreasing on $]x_0, \infty[$. Now, $f'(0) = a_1 > 0$. Furthermore, $\sum_{i=0}^N (i+1)a_{i+1} x^i = o(x_{N+2})$ and the coefficient of the term of order $N+2$ in $f'$ is negative, just as the rest of the series — except the coefficient of order $N+1$, which might also be zero. Hence, $\lim_{x \to \infty} f'(x) = -\infty$. Hence, there exists a unique $y_0 > x_0$ such that $f'$ is positive on $[0, y_0[$ and negative on $]y_0, \infty[$. Thus, $f$ is increasing on $[0, y_0[$ and decreasing on $]y_0, \infty[$.
\end{proof}

\end{document}